\documentclass[a4paper,UKenglish,cleveref, autoref, thm-restate]{lipics-v2021}

\usepackage[T1]{fontenc}
\usepackage{graphicx}
\usepackage{amsmath}
\usepackage{amssymb}
\usepackage{amsfonts}
\usepackage{comment}
\usepackage{xspace}

\usepackage{thm-restate}

\usepackage{hyperref}
\usepackage{here}

\usepackage{cite}
\usepackage{enumerate}

\nolinenumbers

\hideLIPIcs

\usepackage[linesnumbered,noend,ruled,vlined]{algorithm2e}

\SetKwInput{KwInput}{Input}     
\SetKwInput{KwOutput}{Output}   
\SetKw{KwTo}{to}

\usepackage[framemethod=tikz]{mdframed}

\global\long\def\R{\mathbb{R}}%
\global\long\def\F{\mathbb{F}}%

\global\long\def\M{\mathcal{M}}%
\global\long\def\I{\mathcal{I}}%

\global\long\def\tO{\tilde{O}}%
\global\long\def\TO{\tilde{O}_{\varepsilon}}%

\global\long\def\bv{\boldsymbol{v}}%
\global\long\def\bw{\boldsymbol{w}}%
\global\long\def\bm{\boldsymbol{m}}%
\global\long\def\bb{\boldsymbol{b}}%
\global\long\def\bu{\boldsymbol{u}}%
\global\long\def\bx{\boldsymbol{x}}%
\global\long\def\by{\boldsymbol{y}}%
\global\long\def\bp{\boldsymbol{p}}%
\global\long\def\bq{\boldsymbol{q}}%
\global\long\def\ba{\boldsymbol{a}}%
\global\long\def\be{\boldsymbol{e}}%
\global\long\def\Wp{W^{\perp}}%

\global\long\def\eps{\varepsilon}%
\global\long\def\epsilon{\varepsilon}%

\global\long\def\nnz{\mathrm{nnz}}%

\usepackage{tabularx}
\usepackage{environ}

\usepackage{cite}
\usepackage{url}

\usepackage[normalem]{ulem}
\usepackage{color}

\newcommand{\ShowComment}{}

\ifdefined\ShowComment
\newcommand{\terao}[1]{{\bf \color{red} COMMENT: #1}}
\else
\newcommand{\terao}[1]{}
\fi

\title{Faster Approximate Linear Matroid Intersection}

\author{Tatsuya Terao}{Research Institute for Mathematical Sciences, Kyoto University \and \url{https://otera99.github.io/}}{ttatsuya@kurims.kyoto-u.ac.jp}{https://orcid.org/0000-0002-3530-2194}{}

\authorrunning{T.\ Terao}

\Copyright{Tatsuya Terao}

\keywords{Linear matroid intersection, fast approximation algorithm}

\date{}

\ccsdesc[500]{Theory of computation~Algorithm design techniques}
\ccsdesc[500]{Mathematics of computing~Matroids and greedoids}

\begin{document}


\acknowledgements{
We thank Yusuke Kobayashi for his insightful comments, generous support, and valuable feedback on the manuscript.
In particular, his input greatly helped us establish Lemma~\ref{lem:span_computation}.
We also thank Kou Hamada for helpful conversations that contributed to the initial motivation for this work.
We are grateful to Takashi Noguchi for his helpful comments, which simplified the proof of Lemma~\ref{lem:span_computation}. 
Finally, we thank the anonymous reviewers for their careful reading and valuable suggestions.
This work was partially supported by the joint project of Kyoto University and Toyota Motor Corporation, titled ``Advanced Mathematical Science for Mobility Society'' and by JSPS KAKENHI Grant Number JP24KJ1494.
}

\maketitle

\begin{abstract}

We consider a fast approximation algorithm for the linear matroid intersection problem. In this problem, we are given two $r \times n$ matrices $M_1$ and $M_2$, and the objective is to find a largest set of columns that are linearly independent in both $M_1$ and $M_2$. We design a $(1 - \varepsilon)$-approximation algorithm with time complexity $\tilde{O}_{\varepsilon}(\mathrm{nnz}(M_1) + \mathrm{nnz}(M_2) + r_{*}^{\omega})$, where $\mathrm{nnz}(M_i)$ denotes the number of nonzero entries in $M_i$ for $i = 1, 2$, $r_{*}$ denotes the maximum size of a common independent set, and $\omega < 2.372$ denotes the matrix multiplication exponent. Our approximation algorithm is faster than the exact algorithm by Harvey [FOCS'06 \& SICOMP'09] and Cheung--Kwok--Lau [STOC'12 \& JACM'13], which runs in $\tilde{O}(\mathrm{nnz}(M_1) + \mathrm{nnz}(M_2) + n r_{*}^{\omega - 1})$ time.

We also develop a fast $(1 - \varepsilon)$-approximation algorithm for the weighted version of the linear matroid intersection problem. In fact, we design a $(1 - \varepsilon)$-approximation algorithm for weighted linear matroid intersection with time complexity $\tilde{O}_{\varepsilon}(\mathrm{nnz}(M_1) + \mathrm{nnz}(M_2) + r_{*}^{\omega})$. Our algorithm improves upon the $(1 - \varepsilon)$-approximation algorithm by Huang--Kakimura--Kamiyama [SODA'16 \& Math. Program.'19], which runs in $\tilde{O}_{\varepsilon}(\mathrm{nnz}(M_1) + \mathrm{nnz}(M_2) + nr_{*}^{\omega - 1})$ time.

To obtain these results, we combine Quanrud's adaptive sparsification framework [ICALP'24] with a simple yet effective method for efficiently checking whether a given vector lies in the linear span of a subset of vectors, which is of independent interest.

\end{abstract}

\newpage

\section{Introduction} \label{sec:intro}

\subparagraph*{Fast Algorithms for Linear Matroid Problems}

Matroids are combinatorial structures that abstract linear independence and provide a unifying framework capturing fundamental properties of matrices and graphs.
We focus on fast algorithms for problems on linear matroids that are given explicitly as the columns of an $r \times n$ matrix $M$ over a field $\F$.

Motivated by the view that running time should reflect not only the matrix dimensions $r \times n$ but also its sparsity, we design faster algorithms whose complexity depends on $\mathrm{nnz}(M)$ (the number of nonzero entries of $M$) in addition to $r$ and $n$.
In fact, many commonly studied matroids admit sparse linear representations.
For example, it is well known that a graphic matroid can be represented by a vertex-edge incidence matrix with exactly two nonzero entries per edge, resulting in a total of $2m$ nonzero entries, where $m$ is the number of edges (see e.g., \cite[Section 5.1]{oxley2006matroid}).

From this perspective, an important contribution is due to Cheung--Kwok--Lau~\cite{cheung2013fast}, who developed a remarkable sketching technique and, as a result, presented a surprising algorithm for computing the rank of a matrix $M$ in $\tilde{O}(\mathrm{nnz}(M) + (\mathrm{rank}(M))^{\omega})$ time.\footnote{The $\tO$ notation omits polylog factors.}
Note that even checking whether a given subset of vectors is independent already requires $O((\mathrm{rank} (M))^{\omega})$ time, and reading the input takes $O(\mathrm{nnz}(M))$ time.

Nguy$\tilde{\hat{\text{e}}}$n~\cite{nguyen2019fast} presented an elegant algorithm for computing a maximum-weight base of a linear matroid represented by a matrix $M$. 
While a direct application of the sketching technique of Cheung--Kwok--Lau would only yield a running time of $\tilde{O}\!\big(\mathrm{nnz}(M)+n \cdot (\mathrm{rank}(M))^{\omega-1}\big)$, his refined use of this technique leads to $\big(\mathrm{nnz}(M) + (\mathrm{rank}(M))^{\omega}\big)^{1+o(1)}$.
This problem has also been studied, under the name column rank profile problem~\cite{dumas2013simultaneous, jeannerod2013rank, storjohann2014linear, dumas2017fast}, which asks for the lexicographically smallest sequence of $\mathrm{rank}(M)$ column indices such that the corresponding columns of the matrix are linearly independent. Dumas--Pernet--Sultan~\cite{dumas2017fast} achieved the same time complexity for the column rank profile problem.

\subparagraph*{Linear Matroid Intersection.} \label{subsec:intro_linear_matroid_intersection}

In this work, we consider fast approximation algorithms for linear matroid intersection.
This linear matroid intersection problem can be described without relying on matroid terminology as follows: Given two $r \times n$ matrices $M_1$ and $M_2$ over a field $\F$, where the columns in $M_1$ and $M_2$ are indexed by $\{1, \ldots, n \}$, the objective is to find a set $I \subseteq \{1, \ldots, n\}$ of maximum size such that the columns indexed by $I$ are linearly independent in both $M_1$ and $M_2$.
Many important problems in combinatorial optimization, such as bipartite matching, rainbow spanning tree, arborescence, and spanning tree packing, can be viewed as special cases of linear matroid intersection.

For the linear matroid intersection problem, Cunningham \cite{cunningham1986improved} gave a combinatorial algorithm with time complexity $\tilde{O}(nr^2)$.
Gabow--Xu \cite{gabow1989efficient, gabow1996efficient} proposed a faster combinatorial algorithm with time complexity $\tilde{O}(nrr_*^{1/(4-\omega)}) = O(nrr_*^{0.62})$, where $\omega$ is the exponent of matrix multiplication, known to satisfy $\omega < 2.371339$ \cite{alman2025more}.
Here, $r_*$ denotes the maximum size of a common independent set.
Harvey \cite{harvey2009algebraic} developed an algebraic algorithm with time complexity $O(nr^{\omega - 1})$.
Cheung--Kwok--Lau~\cite{cheung2013fast} combined their elegant sketching technique with Harvey’s algorithm~\cite{harvey2009algebraic}, obtaining an algorithm with running time $\tilde{O}(\mathrm{nnz}(M_1) + \mathrm{nnz}(M_2) + nr_*^{\omega - 1})$.
In fact, simply applying their sketching technique reduces an $r \times n$ input matrix to an $O(r_*) \times n$ matrix while preserving linear dependence among the columns.


Given the progress on the case of a single matroid \cite{cheung2013fast, dumas2017fast, nguyen2019fast}, it is natural to ask whether linear matroid intersection can be solved in $\tilde{O}(\mathrm{nnz}(M_1) + \mathrm{nnz}(M_2) + r_*^{\omega})$ time.
Note that even checking whether a given subset is independent in a single matroid already requires $O(r_*^{\omega})$ time, and reading the input takes $O(\mathrm{nnz}(M_1) + \mathrm{nnz}(M_2))$ time.
Avoiding the dependence on $n$ is most meaningful when $r_*\ll n$, a regime that arises naturally, for example, for graphic matroids.
However, achieving a running time of $\tilde{O}(\mathrm{nnz}(M_1) + \mathrm{nnz}(M_2) + r_*^{\omega})$ seems quite challenging.
To achieve such a running time, one must avoid explicitly computing the product of an $O(r_*) \times n$ matrix and an $n \times O(r_*)$ matrix, which would already take $O(n r_*^{\omega - 1})$ time.
Here, it is well known that $M_1$ and $M_2$ have a common base if and only if $\det (M_1 D M_2^\top) \neq 0$, where $D$ is a  diagonal matrix with variables on the diagonal $D_{i, i} = x_i$ (see for e.g., \cite[Lemma 2.5]{gurjar2020linear}).
Using this matrix formulation, a standard approach to the decision version of linear matroid intersection, i.e., deciding whether a common base exists, requires computing the product of an $O(r_*)\times n$ matrix and an $n\times O(r_*)$ matrix.
Hence, to our knowledge, no prior result has shown that the time complexity for merely estimating the size of a maximum common independent set can be faster than that of Cheung--Kwok--Lau.

\subparagraph*{Our Result for Maximum Cardinality Linear Matroid Intersection.} \label{subsec:intro_our_result_cardinality}



Our main result is to develop a $(1-\eps)$-approximation algorithm for linear matroid intersection running in $\tilde{O}_{\eps}(\mathrm{nnz}(M_1) + \mathrm{nnz}(M_2) + r_*^{\omega})$ time.\footnote{The $\tilde{O}_{\epsilon}$ notation omits polylog factors and polynomial factors in $\epsilon^{-1}$.}
We note that our algorithm is applicable over any field.

\begin{restatable}{theorem}{cardinalityalgo} \label{thm:cardinality_main_thm}
    For any $\eps > 0$, a $(1 - \eps)$-approximate solution to the maximum cardinality linear matroid intersection problem can be computed with high probability\footnote{We say ``with high probability'' if the failure probability is at most $n^{-\Theta(1)}$.} in $\tilde{O}_{\eps}(\mathrm{nnz}(M_1) + \mathrm{nnz}(M_2) + r_*^{\omega})$ time.
\end{restatable}


Our approximation algorithm in Theorem \ref{thm:cardinality_main_thm} is faster than the exact algorithm obtained by combining Harvey's algorithm~\cite{harvey2009algebraic} with the sketching technique of Cheung--Kwok--Lau~\cite{cheung2013fast}, which runs in $\tilde{O}(\mathrm{nnz}(M_1) + \mathrm{nnz}(M_2) + nr_*^{\omega - 1})$ time.
Even when $M_1$ and $M_2$ are dense and $r_* = \Theta(r)$, our approximation algorithm achieves a running time of $\tilde{O}_{\epsilon}(nr + r^{\omega})$, which is faster than the $O(nr^{\omega - 1})$ time required by Harvey's exact algorithm~\cite{harvey2009algebraic}.
Surprisingly, our complexity is almost the same as that of computing the rank of a matrix~\cite{cheung2013fast}, and also that of finding a maximum-weight base of a single linear matroid~\cite{dumas2017fast, nguyen2019fast}.



\subparagraph*{Weighted Linear Matroid Intersection} \label{subsec:weighted_linear_matroid_intersection}

We also consider the weighted version of the linear matroid intersection problem.
In this problem, we are given two $r \times n$ matrices $M_1$ and $M_2$ over a field $\F$, whose columns are indexed by $\{1, \ldots, n \}$, and a weight function $c : \{ 1, \ldots, n\} \to \mathbb{R}_{\geq 0}$.
The objective is to find a set $I \subseteq \{1, \ldots, n\}$ such that the columns indexed by $I$ are linearly independent in both $M_1$ and $M_2$, and the total weight $\sum_{e \in I} c(e)$ is maximized.
Gabow--Xu~\cite{gabow1989efficient, gabow1996efficient} gave a combinatorial algorithm with time complexity $\tilde{O}(nr^{\frac{7 - \omega}{5 - \omega}} \log(nC)) = \tilde{O}(nr^{1.77} \log(nC))$, where $C$ is the largest given weight.
Harvey~\cite{harvey2007algebraic} proposed an algebraic algorithm with time complexity $\tilde{O}(C^{1 + o(1)} nr^{\omega - 1})$.
Huang--Kakimura--Kamiyama~\cite{huang2016exact} presented a $(1 - \eps)$-approximation algorithm with time complexity $\tilde{O}_{\eps}(\mathrm{nnz}(M_1) + \mathrm{nnz}(M_2) + nr_{*}^{\omega - 1})$, where $r_{*}$ is the maximum size of a common independent set.

\subparagraph*{Our Result for Maximum Weight Linear Matroid Intersection.} \label{subsec:intro_our_result_weight}

Using a similar idea to that used to obtain Theorem \ref{thm:cardinality_main_thm}, we also develop a $(1-\eps)$-approximation algorithm for weighted linear matroid intersection running in $\tilde{O}_{\eps}(\mathrm{nnz}(M_1) + \mathrm{nnz}(M_2) + r_{*}^{\omega})$ time.
We note that our algorithm is applicable over any field.

\begin{restatable}{theorem}{weightedalgo}  \label{thm:weighted_main_thm}
    For any $\eps > 0$, a $(1 - \eps)$-approximate solution to the weighted linear matroid intersection problem can be computed with high probability in $\tilde{O}_{\eps}(\mathrm{nnz}(M_1) + \mathrm{nnz}(M_2) + r_{*}^{\omega})$ time.
\end{restatable}


Our algorithm in Theorem \ref{thm:weighted_main_thm} improves upon the $(1-\eps)$-approximation weighted linear matroid intersection algorithm by Huang--Kakimura--Kamiyama~\cite{huang2016exact}, which runs in $\tilde{O}_{\eps}(\mathrm{nnz}(M_1) + \mathrm{nnz}(M_2) + nr_{*}^{\omega - 1})$ time.

\subsection{Technical Overview} \label{subsec:intro_technical_overview}

\subparagraph*{Adaptive Sparsification Framework of Quanrud~\cite{quanrud:LIPIcs.ICALP.2024.118}.}

Our linear matroid intersection algorithm builds upon the adaptive sparsification framework of Quanrud~\cite{quanrud:LIPIcs.ICALP.2024.118},
which was developed to obtain a fast $(1-\eps)$-approximation algorithm for matroid intersection in the oracle model.
However, we emphasize that a direct application of Quanrud's framework does not yield the running time achieved by our algorithm.

In Quanrud's framework, the original problem over $n$ elements is reduced to $O(\log (n) / \eps)$ primal-dual instances of matroid intersection, each over $\tilde{O}_{\eps}(r_*)$ elements, where $r_*$ is the maximum size of a common independent set.
For linear matroid intersection, computing primal and dual solutions to sparsified instances can be done straightforwardly using Harvey's exact algorithm for linear matroid intersection~\cite{harvey2009algebraic}.
Applied to the sparsified instance of size $\tilde{O}_{\eps}(r_*)$, Harvey's algorithm runs in $\tilde{O}_{\eps}(r_*^{\omega})$ time, which is sufficiently fast.
To adaptively construct sparsified instances, it is necessary to compute $\mathrm{span}_{\M_1}(\tilde{S})$ and $\mathrm{span}_{\M_2}(\tilde{T})$, where $(\tilde{S}, \tilde{T})$ denotes a dual solution to the sparsified instance.
For linear matroids, a straightforward method can compute the spans in $O(nr_*^{\omega - 1})$ time, but this is still too slow for our purpose.

\subparagraph*{Challenge in Computing a Span.}

Without using matroid terminology, our task is as follows: Given an $r \times n$ matrix $M$ and a subset $S \subseteq [n]$, determine whether each column vector $\boldsymbol{m}_i$ lies in the linear span of the set $\{ \boldsymbol{m}_j  \mid j \in S\}$, where $\boldsymbol{m}_1, \ldots, \boldsymbol{m}_n$ denote the vectors corresponding to the columns of $M$.
We present a simple randomized algorithm that solves this task in $\tilde{O}(\mathrm{nnz}(M) + r^{\omega})$ time, which is of independent interest.

In the oracle model, where an algorithm accesses a matroid $\M$ through an independence oracle, the span 
\begin{equation*}
\mathrm{span}_{\M}(S) = \{ v \in V \mid {\rm rank}_{\M}(S \cup \{v\}) = {\rm rank}_{\M}(S) \}
\end{equation*}
of a set $S$ can be  computed naively using a linear number of queries.
However, when the matroid is given via a matrix representation, the standard approaches are too inefficient for our purpose.
Naively, each independence query takes $O(r^{\omega})$ time, and thus the most straightforward implementation requires $O(nr^{\omega})$ time.
The standard Gaussian elimination approach takes $O(nr^2)$ time to compute the span.
Using the fact that the linear representation of $\M / S$, the matroid obtained by contracting $\M$ by $S$, can be computed in $O(n r^{\omega - 1})$ time (see \cite[Fact 4.1]{harvey2009algebraic}), we can reduce the time complexity to $O(n r^{\omega - 1})$.
However, this complexity is still inefficient for obtaining a linear matroid intersection algorithm that outperforms the best known exact algorithm.
To compute the span faster than $O(n r^{\omega - 1})$ time, one must avoid computing the product of an $r \times n$ matrix and an $n \times r$ matrix, which is a challenging task.

To overcome this difficulty, our fast span computation algorithm shares a similar spirit with the famous Freivalds' algorithm \cite{freivalds1979fast} (see e.g., \cite[Section 7.1]{motwani2013randomized}) for verifying matrix multiplication without actually performing the multiplication.
In our approach, we sample a vector $\boldsymbol{v}$ from the orthogonal complement of $\mathrm{span} \{ \boldsymbol{m}_j  \mid j \in S\}$, which can be computed in $\tilde{O}(\mathrm{nnz}(M) + r^{\omega})$ time.
Then, for each column vector $\boldsymbol{m}_i$, we check whether $\boldsymbol{v}^\top \boldsymbol{m}_i = 0$.
It is obvious that if $\boldsymbol{m}_i \in \mathrm{span} \{ \boldsymbol{m}_j \mid j \in S \}$, then $\boldsymbol{v}^\top \boldsymbol{m}_i = 0$ always holds.  
Furthermore, we can show that if $\boldsymbol{m}_i \notin \mathrm{span} \{ \boldsymbol{m}_j \mid j \in S \}$, then $\boldsymbol{v}^\top \boldsymbol{m}_i \neq 0$ with high probability.
The total time required to compute all inner products $\boldsymbol{v}^\top \boldsymbol{m}_i$ for all $i \in [n]$ is $O(\mathrm{nnz}(M))$. 
Therefore, the total running time of our span computation algorithm is $\tilde{O}(\mathrm{nnz}(M) + r^{\omega})$, which is efficient enough to enable the design of a linear matroid intersection algorithm that outperforms the best known exact algorithm.

\subparagraph*{Challenge for Weighted Linear Matroid Intersection.}

For weighted linear matroid intersection, a naive combination of the adaptive sparsification framework of Quanrud~\cite{quanrud:LIPIcs.ICALP.2024.118} and our fast span computation algorithm yields a $(1 - \eps)$-approximation algorithm with time complexity $\tilde{O}_{\eps}(\mathrm{nnz}(M_1) + \mathrm{nnz}(M_2) + r_{*}^{\omega + 1})$.
This additional factor arises because, in the weighted setting, we need to test membership in $\operatorname{span}_{\mathcal{M}}(S)$ not for a single candidate set $S$, but for $O(r_*)$ different candidates.
To improve the running time to $\tilde{O}_{\varepsilon}(\mathrm{nnz}(M_1) + \mathrm{nnz}(M_2) + r_{*}^{\omega})$, we carefully employ the fast Gram--Schmidt orthonormalization algorithm of van den Brand~\cite{van2021unifying}.
Concretely, we use van den Brand's algorithm to compute, for each set in the support of the dual solution in the weighted setting, a basis for the corresponding subspace as well as a basis for its orthogonal complement.

When considering matrices over the real field $\R$, the idea of employing the fast Gram–Schmidt orthonormalization algorithm of van den Brand works correctly in conjunction with the techniques described above.
However, over an arbitrary field $\mathbb{F}$, the argument may fail.
This is because, when considering matrices over a field other than $\R$, it may happen that a linear subspace $W$ satisfies $W \cap \Wp \neq \{ \boldsymbol{0} \}$ and $W + \Wp \neq \F^r$,\footnote{For example, consider $W = \mathrm{span}\{(1, 1, 0, 0)^\top, (0, 0, 1, 1)^\top\}$ over the binary field $\mathrm{GF}(2)$. Then, $\Wp = \mathrm{span}\{(1, 1, 0, 0)^\top, (0, 0, 1, 1)^\top\}$. In this case, neither $W \cap \Wp = \{ \boldsymbol{0} \}$ nor $W + \Wp = (\mathrm{GF}(2))^4$ holds.} where $W^\perp$ is the orthogonal complement of $W$ (i.e., $W^\perp = \{ \boldsymbol{v} \in \mathbb{F}^r \mid \forall \boldsymbol{w} \in W, \boldsymbol{v}^\top \boldsymbol{w} = 0  \}$).
This issue matters in our setting because we need to maintain bases for subspaces corresponding to many different sets, and for this purpose, our algorithm relies on the property that $W \cap \Wp = \{ \boldsymbol{0} \}$ and $W \oplus \Wp = \F^r$.

To design an algorithm applicable to any field, we replace the Euclidean inner product $\bv^\top \bw$ with the bilinear form $\langle \boldsymbol{v}, \boldsymbol{w} \rangle = \bv^\top B \bw$, where we choose $w_1, \ldots, w_r$ independently and uniformly at random from $\mathbb{F}$, and define a diagonal $r \times r$ matrix $B$ such that $B_{i, i} = w_i$.
Here, we define the orthogonal complement with respect to the bilinear form $\langle \boldsymbol{v}, \boldsymbol{w} \rangle = \bv^\top B \bw$ as $\Wp_B = \{ \boldsymbol{v} \in \mathbb{F}^r \mid \forall \boldsymbol{w} \in W, \langle \boldsymbol{v}, \boldsymbol{w} \rangle = \bv^\top B \bw = 0  \}$.
For a fixed linear subspace $W$, we can show that both $W \cap \Wp_B = \{ \boldsymbol{0} \}$ and $W \oplus \Wp_B = \mathbb{F}^r$ hold with high probability.

To determine whether $\boldsymbol{m}_i \in \mathrm{span} \{ \boldsymbol{m}_j \mid j \in S \}$, it suffices to check whether $\langle \boldsymbol{v}, \boldsymbol{m}_i \rangle = \bv^\top B \bm_i = 0$.
Here, $\bv$ is a vector sampled from the orthogonal complement of $\mathrm{span} \{ \boldsymbol{m}_j \mid j \in S \}$ with respect to the bilinear form $\langle \boldsymbol{v}, \boldsymbol{w} \rangle = \bv^\top B \bw$.
We make use of a fast orthogonalization algorithm with respect to the bilinear form $\langle \bv, \bw \rangle = \bv^\top B \bw$, which is a modified version of the fast Gram--Schmidt orthonormalization algorithm of van den Brand~\cite{van2021unifying}.


\subsection{Related Work} \label{subsec:additional_related_work}

\subparagraph*{Faster Matroid Intersection under Oracle Models.} \label{subsec:intro_matroid_intersection}

In addition to the linear matroid model, another standard model for handling matroids is the oracle model.
In this model, an algorithm accesses a matroid through an oracle that answers queries about its structure.
In recent years, the development of faster algorithms for matroid intersection under oracle models has attracted significant attention.
Starting the work of Edmonds \cite{edmonds1970submodular, edmonds1979matroid}, many fast algorithms in the oracle model have been studied~\cite{aigner1971matching, tomizawa1974algorithm, lawler1975matroid, cunningham1986improved, fujishige1995efficient, shigeno1995dual, lee2015faster, chekuri2016fast, huang2016exact, nguyen2019note, chakrabarty2019faster, blikstad2021breaking_STOC, blikstad2021breaking, u2022subquadratic, blikstad2023fast,quanrud:LIPIcs.ICALP.2024.118,terao2024deterministic,blikstad2024efficient, dudeja2026weighted}.
There have also been many studies on fast matroid intersection algorithms for several important specific classes of matroids~\cite{gabow1979efficient, gabow1985efficient, cunningham1986improved, gabow1989efficient, gabow1996efficient, xu1994fast, harvey2007algebraic, harvey2009algebraic,huang2016exact}.
Recently, Blikstad--Mukhopadhyay--Nanongkai--Tu \cite{blikstad2023fast} proposed a remarkable new oracle model, called dynamic oracle, which unifies various results on fast matroid intersection algorithms for specific classes of matroids; see \cite[Table 2]{blikstad2023fast}.

Note that the recent dynamic oracle framework of Blikstad--Mukhopadhyay--Nanongkai--Tu \cite{blikstad2023fast} yields a linear matroid intersection algorithm with time complexity $O(n^{2.5925}\sqrt{r_*})$, using the dynamic matrix rank maintenance algorithm of \cite{van2019dynamic}.
However, this complexity is slower than that of Harvey \cite{harvey2009algebraic}.

\subparagraph*{Additional Related Work.}
There are several studies on fast algebraic algorithms for problems that generalize linear matroid intersection, such as linear matroid parity~\cite{cheung2014algebraic}, fractional linear matroid parity~\cite{oki2025algebraic}, and linear delta-matroid parity~\cite{koana_et_al:LIPIcs.STACS.2025.62}.
Matoya--Oki~\cite{matoya2022pfaffian} derandomized Harvey's \cite{harvey2009algebraic} linear matroid intersection algorithm for a special case known as a Pfaffian pair.
Linear matroid intersection has also been well studied in parallel and other computational models~\cite{narayanan1994randomized, gurjar2020linear, gurjar_et_al:LIPIcs.ESA.2024.63, agarwala_et_al:LIPIcs.ITCS.2026.3}.
The intersection of multiple linear matroids has also attracted interest in the parameterized algorithm community; see, for example, \cite[Section 12.3]{cygan2015parameterized} and \cite{barvinok1995new, marx2009parameterized, brand2021parameterized, eiben2024determinantal}.


\subparagraph*{Concurrent Work.}

Independent of our work, Dudeja--Grilnberger~\cite{dudeja2026weighted} presented a remarkable reduction in the independence oracle model that converts any $(1-\eps)$-approximation algorithm for unweighted matroid intersection into a $(1-\eps)$-approximation algorithm for weighted matroid intersection, while increasing the running time by only a factor of $\log C$, where $C$ denotes the aspect ratio. 
Thus, one might hope that, when combined with our unweighted result in Theorem~\ref{thm:cardinality_main_thm}, their reduction could yield a weighted result comparable to Theorem~\ref{thm:weighted_main_thm}.
However, their reduction relies on several subroutines, including a procedure for greedily computing a common independent set, and it is unclear whether all of these subroutines can be implemented efficiently for linear matroids.
Accordingly, it is unclear whether their reduction, combined with Theorem~\ref{thm:cardinality_main_thm}, can yield a result comparable to Theorem~\ref{thm:weighted_main_thm}.

\section{Preliminaries}

\subsection{Basic Notation}

The set $\{ 1, \ldots, n\}$ is denoted by $[n]$.
If $U$ is a set, then $\binom{U}{k}$ denotes the set of all its subsets of size $k$.

Consider a matrix $M$ over a field $\mathbb{F}$.
Let $\mathrm{nnz}(M)$ denote the number of nonzero entries in $M$.
Let $M[S]$ denote the submatrix consisting of columns indexed by $S$.
Let $M[R, C]$ denote the submatrix consisting of the rows indexed by $R$ and the columns indexed by $C$.
Let $[A \mid B]$ denote the matrix whose columns consist of those of $A$, followed by those of $B$.
We write $I$ to denote the identity matrix. 
Let $\mathrm{span}\{ \bv_1, \ldots, \bv_k \}$ denote the linear span of the set of vectors $\{ \bv_1, \dots, \bv_k \}$.

\subsection{Assumptions and Conventions}


When $\mathbb{F}$ is a finite field, we can assume that $|\mathbb{F}| = \Omega(\mathrm{poly}(n))$ by the following well-known lemma.
As a result of this assumption, the running time of our algorithm increases by a multiplicative factor of $\mathrm{polylog}(n)$.

\begin{lemma}[from {\cite[Lemma 2.1]{cheung2013fast}}] \label{lem:matrix_elem_assumption}
    Let $A$ be an $r \times n$ matrix over a field $\mathbb{F}$ with $p^c$ elements.
    We can construct a finite field $\mathbb{F}'$ with $p^{ck} = \Omega(\mathrm{poly}(n))$ elements and an injective mapping $f\colon\mathbb{F} \to \mathbb{F}'$ so that the image of $\mathbb{F}$ is a subfield of $\mathbb{F}'$.
    Then, the $r \times n$ matrix $A'$ where $A'_{ij} = f(A_{ij})$ satisfies the following property: a set of columns of $A$ is independent if and only if the corresponding set of columns of $A'$ is independent. 
    This preprocessing step takes $O(\mathrm{nnz}(A))$ time. Each field operation in $\mathbb{F}'$ can be performed using $\mathrm{polylog}(n)$ field operations in $\mathbb{F}$.
\end{lemma}

When $\mathbb{F}$ is an infinite field, we assume the exact arithmetic model where each field operation is performed in one unit of time.
In the algorithm, we also require the ability to sample a random element from an arbitrary subset of size $\Omega(\mathrm{poly}(n))$ in $\mathbb{F}$.
This assumption is needed for the application of the Schwartz--Zippel Lemma.\footnote{The Schwartz--Zippel Lemma is used in Harvey's exact algorithm for linear matroid intersection~\cite{harvey2009algebraic} and in the sketching technique of Cheung--Kwok--Lau~\cite{cheung2013fast}, both of which are components of our algorithm.}

\begin{lemma}[Schwartz--Zippel Lemma {\cite{schwartz1980fast}}] \label{lem:schwartz_zippel}
    Let $P \in \F[x_1, \ldots, x_n]$ be a nonzero polynomial of total degree $d$ over a field $\mathbb{F}$.
    Let $S$ be a finite subset of $\F$ and let $r_1, \ldots, r_n$ be chosen independently and uniformly at random from $S$.
    Then, the probability that $P(r_1, \ldots, r_n) = 0$ is at most $d / |S|$.
\end{lemma}

We assume that two $n \times n$ matrices can be multiplied in $O(n^{\omega})$ time.
The current best upper bound on $\omega$ is $\omega < 2.371339$ \cite{alman2025more}.
For an $n \times n$ matrix, it is known that computing the determinant, rank, and inverse can all be done in the same time complexity as one matrix multiplication \cite{bunch1974triangular, harvey2008matchings}.

The assumptions in this subsection were also made in \cite{cheung2013fast, nguyen2019fast}.

\subsection{Matroid Preliminaries}

For the basics of matroid theory, see, for example, the excellent textbooks of Oxley \cite{oxley2006matroid} or Schrijiver \cite{schrijver2003combinatorial}.

A pair $\M = (V, \mathcal{I})$ of a finite set $V$ and a non-empty set family $\mathcal{I} \subseteq 2^{V}$ is called a {\em matroid} if the following properties are satisfied.
    
    \begin{description}
        \item[(Downward closure property)] 
        If $I \in \mathcal{I}$ and $J \subseteq I$, then $J \in \mathcal{I}$.
        \item[(Augmentation property)]
        If $I, J \in \mathcal{I}$ and $|J| < |I|$, then there exists $v \in I \setminus J$ such that $J \cup \{v\} \in \mathcal{I}$.
    \end{description}

A set $I \subseteq V$ is called {\em independent} if $I \in \mathcal{I}$ and {\em dependent} otherwise.

\subparagraph*{Matroid Rank and Span.}
Motivated by linear algebra, matroids have associated notions of a \emph{rank function} and a \emph{span function}, defined as follows.
For a matroid $\M = (V, \mathcal{I})$, the {\em rank} of $\M$ is ${\rm rank}(\M) = \max \{ |I| \mid I \in \mathcal{I} \}$.
In addition, for any $S \subseteq V$, the {\em rank} of $S$ is ${\rm rank}_{\M}(S) = \max \{ |I| \mid I \subseteq S, I \in \mathcal{I} \}$.
Moreover, for any $S \subseteq V$, the {\em span} of $S$ is ${\rm span}_{\M}(S) = \{ v \in V \mid {\rm rank}_{\M}(S \cup \{v\}) = {\rm rank}_{\M}(S) \}$.

\subparagraph*{Matroid Intersection.}
In the {\em matroid intersection} problem, we are given two matroids $\mathcal{M}_1 = (V, \mathcal{I}_1)$ and $\mathcal{M}_2 = (V, \mathcal{I}_2)$, and the objective is to find a common independent set $I \in \I_1 \cap \I_2$ of maximum size.

Edmonds~\cite{edmonds1970submodular} proved the following min-max theorem, which is now well known.

\begin{theorem}[Matroid intersection theorem {\cite{edmonds1970submodular}}; see also {\cite[Theorem 41.1]{schrijver2003combinatorial}}] \label{fact:matroid_intersection_minmax}

Let $\M_1 = (V, \I_1)$, $\M_2 = (V, \I_2)$ be matroids, with rank function $\mathrm{rank}_{\M_1}$ and $\mathrm{rank}_{\M_2}$, respectively. 
Then, 

\begin{equation*}
    \max_{I \in \I_1 \cap \I_2} |I| = \min_{S, T \subseteq V, S \cup T = V} (\mathrm{rank}_{\M_1}(S) + \mathrm{rank}_{\M_2}(T)).
\end{equation*}
\end{theorem}

Here, $(S, T)$ is called a dual solution to matroid intersection, where $S$ and $T$ minimize the right-hand side of the equation in Theorem~\ref{fact:matroid_intersection_minmax}.

In our linear matroid intersection algorithm, we use Harvey's exact algorithm for linear matroid intersection~\cite{harvey2009algebraic} to compute primal and dual solutions to sparsified instances.

\begin{theorem}[Harvey's exact linear matroid intersection algorithm; see {\cite[Section 4.6]{harvey2009algebraic}}] \label{thm:Harvey}
    An exact solution to the maximum cardinality linear matroid intersection problem can be computed with high probability in $\tilde{O}(nr^{\omega - 1})$ time.
    Furthermore, an exact dual solution can also be computed with high probability in $\tilde{O}(nr^{\omega - 1})$ time.
\end{theorem}

\subsection{Sketching Technique of Cheung--Kwok--Lau} \label{subsec:sketching_technique}

In our algorithm, we use the sketching technique of Cheung--Kwok--Lau~\cite{cheung2013fast}.
Using this technique, we can reduce the dimension of the input matrix from $r \times n$ to $O(r_*) \times n$.

\begin{lemma}[Sketching technique of Cheung--Kwok--Lau~{\cite{cheung2013fast}}; from {\cite[Lemma 2.9]{cheung2013fast}}] \label{lem:sketching_tech}
    There is an algorithm that, given an $r \times n$ matrix $A$ over a field $\F$,  returns an $O(k) \times n$ matrix $A'$ over a field $\F$ in $O(\mathrm{nnz}(A))$ time, with $\mathrm{nnz}(A') = O(\mathrm{nnz}(A))$, such that for any set $S \subseteq [n]$ of size at most $k$, if the columns indexed by $S$ are linearly independent in $A$, then the columns of $A'$ indexed by $S$ are also linearly independent with probability at least $1 - O(1/n^{1/3})$.
\end{lemma}

In our fast span computation algorithm, we use the following lemma.

\begin{lemma}[Fast computation of a basis of the orthogonal complement of the column space] \label{lem:orthogonal_complement_basis}
    There is a randomized algorithm that, given an $r \times n$ matrix $M$ over a field $\F$, finds a basis of the orthogonal complement of the column space of $M$ with high probability in $O((\mathrm{nnz}(M) + r^{\omega}) \log n)$ time.
\end{lemma}

Lemma \ref{lem:orthogonal_complement_basis} can be easily derived from the following lemma, which was obtained using the sketching technique of Cheung--Kwok--Lau.
Here, the null space of a matrix $A$ is the subspace of vectors $\bx$ such that $A \boldsymbol{x} = \boldsymbol{0}$.

\begin{lemma}[from {\cite[Theorem 1.3(2)]{cheung2013fast}}; see also {\cite[Section 4.1.2]{cheung2013fast}}] \label{lem:ckl_null_space_basis}
     There is a randomized algorithm that, given an $m \times n$ matrix $A$ over a field $\F$, finds a basis of the null space of $A$ with high probability in $O((\mathrm{nnz}(A) + n (\mathrm{rank}(A))^{\omega - 1}) \log n)$ time.
\end{lemma}

Let $W$ denote the column space of an $r \times n$ matrix $M$, i.e., $W = \mathrm{span} \{ \bm_1, \ldots, \bm_n \}$, where $\boldsymbol{m}_1, \ldots, \boldsymbol{m}_n$ are the column vectors of $M$.
The orthogonal complement of $W$, denoted by $\Wp$, is defined as $W^\perp = \{ \bx \in \mathbb{F}^r \mid \forall \boldsymbol{w} \in W, \boldsymbol{w}^\top \boldsymbol{x} = 0   \}$.
Now, $\bx \in \Wp$ if and only if $\bm_i^\top \bx = 0$ for any $i \in [n]$.
We have $\Wp = \{ \bx \in \F^r \mid M^\top \bx = \boldsymbol{0} \}$.
Therefore, to find a basis of the orthogonal complement of the column space of $M$, we simply compute a basis of the null space of $M^{\top}$ using Lemma \ref{lem:ckl_null_space_basis}.
Note that $M^{\top}$ is an $n \times r$ matrix.
Hence, Lemma~\ref{lem:orthogonal_complement_basis} follows from Lemma~\ref{lem:ckl_null_space_basis}.



\subsection{Fast Gram--Schmidt Orthonormalization} \label{subsec:fast_orthogonal_basis}


For our maximum weight linear matroid intersection algorithm for matrices over $\R$, we use the fast Gram--Schmidt orthonormalization algorithm of van den Brand~\cite{van2021unifying}.


\begin{lemma}[Fast Gram--Schmidt orthonormalization; from van den Brand~{\cite[Section 4.5]{van2021unifying}}] \label{lem:vandenbrand_qr_decomp}
    Given an $r \times n$ matrix $M$ with $r \leq n$ over $\R$, where $\boldsymbol{m}_1, \ldots, \boldsymbol{m}_n$ are column vectors of $M$, there is an algorithm that repeatedly computes $\boldsymbol{m}'_i = \boldsymbol{m}_i - P_{i - 1} \boldsymbol{m}_i$, where $P_{i - 1} \boldsymbol{m}_i$ denotes the orthogonal projection of $\boldsymbol{m}_i$ onto $\mathrm{span}\{ \boldsymbol{m}_1, \ldots, \boldsymbol{m}_{i - 1} \}$, and then normalizes $\bm'_i$ if $\bm'_i \neq \boldsymbol{0}$.
    The total running time of this process is $O(n^\omega)$.
\end{lemma}

To extend our algorithm to arbitrary fields, we modify a fast Gram--Schmidt orthonormalization algorithm of van den Brand into a fast orthogonalization algorithm with respect to the bilinear form $\langle \bv, \bw \rangle = \bv^\top B \bw$; see Section \ref{appendix:general_orthogonalization} for details.

\section{Fast Span Computation Algorithm} \label{thm:span_computation}

As part of our technical contribution, we develop a fast span computation algorithm, which is used as a subroutine in our maximum cardinality linear matroid intersection algorithm. 
The following lemma states its running time guarantee.

\begin{restatable}{lemma}{spanalgo} \label{lem:span_computation}
     There is a randomized algorithm that, given an $r \times n$ matrix $M$ over a field $\F$ whose columns are the representation of elements of a matroid $\M$, and a set $S \subseteq [n]$,\footnote{Although in our application the size of $S$ is always $\tilde{O}_{\eps}(r)$, we state the lemma in a general form for potential broader applicability.}, computes $\mathrm{span}_{\M}(S)$ with high probability in $\tO(\mathrm{nnz}(M) + r^{\omega})$ time.
\end{restatable}

We note that our algorithm is applicable over any field.



The core idea of our algorithm is very simple.
In our algorithm, we first sample a vector $\boldsymbol{v}$ from the orthogonal complement of $\mathrm{span} \{ \boldsymbol{m}_j  \mid j \in S\}$, where $\boldsymbol{m}_1, \ldots, \boldsymbol{m}_n$ denote the vectors corresponding to the columns of the input matrix $M$.
Then, we output the set $\{ i \in [n] \mid \boldsymbol{v}^\top \boldsymbol{m}_i = 0 \}$ as $\mathrm{span}_{\M}(S)$.
See Algorithm \ref{alg:span_computation} for the pseudocode of our algorithm.

Let $W = \mathrm{span} \{ \boldsymbol{m}_j  \mid j \in S\}$.
The orthogonal complement of $W$, denoted by $\Wp$, is defined as $W^\perp = \{ \boldsymbol{v} \in \mathbb{F}^r \mid \forall \boldsymbol{w} \in W, \boldsymbol{v}^\top \boldsymbol{w} = 0  \}$.
When considering matrices over a field other than $\mathbb{R}$, it may occur that $W$ does not satisfy $W \cap W^{\perp} = \{ \boldsymbol{0} \}$ or $W + W^{\perp} = \mathbb{F}^r$.
Nevertheless, for our purpose it suffices that $(W^{\perp})^{\perp} = W$, which in fact holds over any field.\footnote{
This can be verified by a simple argument in linear algebra.
Now, fix any $\bx \in W$.
For any $\by \in \Wp$, we have $\by^\top \bx = 0$.
Thus, $\boldsymbol{x} \in (W^{\perp})^{\perp}$, and hence $W \subseteq (W^{\perp})^{\perp}$.
Moreover, the Rank--nullity theorem implies that $\dim(W^{\perp}) = r - \dim(W)$ and $\dim((W^{\perp})^{\perp}) = r - \dim(W^{\perp})$.
Therefore, $\dim((W^{\perp})^{\perp}) = \dim(W)$, which, together with $W \subseteq (W^{\perp})^{\perp}$, yields $(W^{\perp})^{\perp} = W$.}

Note that we assume $|\mathbb{F}| = \Omega(\mathrm{poly}(n))$.
Recall that, by Lemma \ref{lem:matrix_elem_assumption}, this assumption can be made without loss of generality.

\begin{algorithm}[t]
    By applying Lemma \ref{lem:orthogonal_complement_basis} to $M[S]$, obtain a basis $\{ \boldsymbol{b}_1, \dots, \boldsymbol{b}_k \}$ of the orthogonal complement of $\mathrm{span} \{ \boldsymbol{m}_j  \mid j \in S\}$. \\
    Let $\displaystyle \boldsymbol{v} = \sum_{j = 1}^{k} r_j \boldsymbol{b}_j$, where $r_1, \ldots, r_k$ are chosen independently and uniformly at random from a finite subset $F \subseteq \F$ with $|F| = \Theta(n^2)$. \\
    Let $T \gets \emptyset$ \\
    \For{$i \in [n]$} {
        \If{$\boldsymbol{v}^{\top} \boldsymbol{m}_i = 0$} {
            $T \gets T \cup \{ i \}$
        }
    }
    \Return{$T$}
    \caption{Span computation algorithm} \label{alg:span_computation}
\end{algorithm}


\begin{proof}
In our algorithm, we first compute a basis $\{ \boldsymbol{b}_1, \ldots, \boldsymbol{b}_{k} \}$ of the orthogonal complement of $\mathrm{span} \{ \boldsymbol{m}_j  \mid j \in S\}$ by applying Lemma \ref{lem:orthogonal_complement_basis} to $M[S]$.
This step takes $O((\mathrm{nnz}(M) + r^{\omega}) \log n)$ time.

Next, we choose $r_1, \ldots, r_k$ independently and uniformly at random from a finite subset $F \subseteq \F$ with $|F| = \Theta(n^2)$, and compute 
\begin{equation*}
    \boldsymbol{v} = \sum_{j = 1}^{k} r_j \boldsymbol{b}_j,
\end{equation*}
which takes $O(r^2)$ time, since $k \leq r$ and each $\boldsymbol{b}_j$ is an $r$-dimensional vector.

We claim that the value of $\boldsymbol{v}^{\top} \boldsymbol{m}_i$ satisfies the following two properties.
Here, $\boldsymbol{m_1}, \ldots, \boldsymbol{m_n}$ denote the vectors corresponding to the columns of the matrix $M$.

\begin{itemize}
    \item If $i \in \mathrm{span}_{\M}(S)$, then we have $\boldsymbol{v}^\top \boldsymbol{m}_i = 0$ for every possible choice of $r_1, \ldots, r_k$, since $\boldsymbol{v}$ lies in the orthogonal complement of $\mathrm{span} \{ \boldsymbol{m}_j  \mid j \in S\}$.
    
    \item If $i \notin \mathrm{span}_{\M}(S)$, then we have $\boldsymbol{v}^\top \boldsymbol{m}_i \neq 0$ with probability at least $1 - 1/|F|$. This can be shown as follows. Consider the multivariate polynomial 
    \begin{equation*}
        P(x_1, \ldots, x_k) = \sum_{j = 1}^{k} x_j \boldsymbol{b}_j^{\top} \boldsymbol{m}_i.
    \end{equation*}

    Here, to show that $P(x_1, \ldots, x_k)$ is a nonzero polynomial, we prove the following claim.\footnote{Over the real number field $\R$, the claim is immediate. However, over an arbitrary field $\F$, one cannot in general rely on $W \cap W^{\perp} = \{\boldsymbol{0}\}$ or $W + W^{\perp} = \mathbb{F}^r$, so some care is required. For completeness, we give a careful proof, which uses only the identity $(W^{\perp})^{\perp} = W$.}
    Recall that we are assuming $i \notin \mathrm{span}_{\M}(S)$.
    
    \begin{claim} \label{claim:hozyo}
        There exists some $j \in [k]$ such that $\boldsymbol{b}_j^{\top} \boldsymbol{m}_i \neq 0$.
    \end{claim}

    \begin{proof}[Proof of Claim~\ref{claim:hozyo}]
        Let $W = \operatorname{span}\{ \boldsymbol{m}_j \mid j \in S \}$.
        Here, we have $\boldsymbol{m}_i \notin W$.
        Suppose, for the sake of contradiction, that $\boldsymbol{b}_j^{\top} \boldsymbol{m}_i = 0$ for all $j \in [k]$, where $\{ \boldsymbol{b}_1, \ldots, \boldsymbol{b}_k \}$ form a basis of $W^{\perp}$.
        Then, $\boldsymbol{m}_i$ is orthogonal to every vector in $W^{\perp}$, which implies that $\boldsymbol{m}_i \in (W^{\perp})^{\perp} = W$, contradicting the assumption $\boldsymbol{m}_i \notin W$.
        Therefore, there exists some $j \in [k]$ such that $\boldsymbol{b}_j^{\top} \boldsymbol{m}_i \neq 0$, which completes the proof. 
    \end{proof}
    By Claim~\ref{claim:hozyo}, $P(x_1, \ldots, x_k)$ is a nonzero polynomial of total degree $1$.
    Since $r_1, \ldots, r_k$ are chosen independently and uniformly at random from a finite subset $F \subseteq \F$, by the Schwartz--Zippel Lemma (Lemma \ref{lem:schwartz_zippel}), the probability that
    \begin{equation*}
        P(r_1, \ldots, r_k) = \sum_{j = 1}^{k} r_j \boldsymbol{b}_j^{\top} \boldsymbol{m}_i = \boldsymbol{v}^{\top} \boldsymbol{m}_i = 0
    \end{equation*}
    is at most $1 / |F|$.
\end{itemize}

Then, we only need to check, for each $i \in [n]$, whether $\boldsymbol{v}^{\top} \boldsymbol{m}_i = 0$.
The total time required to compute all inner products $\boldsymbol{v}^\top \boldsymbol{m}_i$ for all $i \in [n]$ is $O(\mathrm{nnz}(M))$. 
Finally, we output the set $\{ i \in [n] \mid \boldsymbol{v}^\top \boldsymbol{m}_i = 0 \}$ as $\mathrm{span}_{\M}(S)$.

There are at most $n$ operations that check whether $i \in \mathrm{span}_{\M}(S)$, and each fails with probability at most $1 / |F|$. 
Therefore, the probability that at least one of them fails is at most $n / |F| = O(1/n)$, where we recall that $|F| = \Theta(n^2)$.

In summary, we can compute $\mathrm{span}_{\M}(S)$ with high probability in $O((\mathrm{nnz}(M) + r^{\omega}) \log n)$ time. 
\end{proof}

\section{Adaptive Sparsification Framework of Quanrud} \label{sec:quanrud}

In this section, we describe the adaptive sparsification framework of Quanrud~\cite{quanrud:LIPIcs.ICALP.2024.118}, which was developed to obtain a fast $(1-\eps)$-approximation algorithm for matroid intersection in the independence oracle model.
Our linear matroid intersection algorithm builds upon this framework.

Quanrud's framework employs a multiplicative weight update method, inspired by Assadi's semi-streaming maximum matching algorithm~\cite{assadi2024simple}.
In this framework, we maintain a weight for each element and repeat the following $O(\log (n) / \varepsilon)$ iterations: In each iteration, we sample $\Theta(\min \{ k \log(n) / \varepsilon, n \})$ elements, in proportion to their weight, compute primal and dual solutions over the sampled elements, and then update the weights of elements.
In Quanrud's paper, $k$ is set to $r_*$, where $r_*$ denote the maximum size of a common independent set.\footnote{
When we apply Quanrud's framework to linear matroid intersection, we set $k = r$ instead of $k = r_*$.
}

\subsection{Maximum Cardinality Matroid Intersection}

In this subsection, we describe Quanrud's framework for maximum cardinality matroid intersection.
The pseudocode is shown in Algorithm \ref{alg:quanrud_cardinality}.

\begin{algorithm}[t]
    Let $w(e) \gets 1$ for all $e \in V$. \\
    \For(\tcp*[f]{$L = O(\log (n) / \eps)$}){$\ell = 1$ to $L$ } {
        Sample a set $V' \subseteq V$ of $\Theta(\min\{ k \log(n) / \eps, n \})$ elements, where each element $e \in V$ is chosen with probability proportional to $w(e)$. \label{linec:sample} \\
        Compute a maximum common independent set $I^{(\ell)}$ and a dual solution $(\tilde{S}^{\ell}, \tilde{T}^{\ell})$ for the subproblem over $V'$. \label{linec:small_computation} \\
        Let $S^{(\ell)} = \mathrm{span}_{\M_1}(\tilde{S}^{(\ell)})$ and $T^{(\ell)} = \mathrm{span}_{\M_2}(\tilde{T}^{(\ell)})$. \label{linec:span_computation} \\
        For each element $e \in S^{(\ell)} \cup T^{(\ell)}$, set $w(e) \gets w(e) / 2$. \label{linec:update_weights}
    }
    Output $I^{(\ell)}$ for some iteration $\ell \in [L]$ for which $|I^{(\ell)}|$ is maximized over all iterations.
    \caption{Quanrud's maximum cardinality matroid intersection algorithm in the independence oracle model (from {\cite[Figure 1]{quanrud:LIPIcs.ICALP.2024.118}})} \label{alg:quanrud_cardinality}
\end{algorithm}

We maintain a weight function $w : V \to \mathbb{R}_{\geq 0}$ over the elements, where all elements initially have equal weights.
We then repeat the following process for $L = O(\log(n) / \eps)$ iterations: 
In each iteration, we first sample $\tilde{O}(k / \eps)$ elements $V' \subseteq V$ in proportion to $w$.
Then, we compute a maximum common independent set $I$ and a dual solution $(\tilde{S}, \tilde{T})$ for the subproblem over $V'$.\footnote{In Quanrud's paper, a $(1 - \varepsilon)$-approximate solution is computed in Line~\ref{linec:small_computation} of Algorithm \ref{alg:quanrud_cardinality}. However, when applying Quanrud's framework to linear matroid intersection, we compute an exact solution. It is clear that this does not affect the correctness.}
Note that $\tilde{S}, \tilde{T} \subseteq V'$.
We extend these sets to their spans, setting $S = \mathrm{span}_{\M_1}(\tilde{S})$ and $T = \mathrm{span}_{\M_2}(\tilde{T})$.
Finally, we update the weights $w$ by setting $w(e) \gets w(e) / 2$ for all elements $e \in S \cup T$.
After $L$ iterations, we output $I$ with the largest cardinality among all iterations.

When applying Quanrud's framework to linear matroid intersection, we use Harvey's linear matroid intersection algorithm~(Theorem \ref{thm:Harvey}) to implement Line~\ref{linec:small_computation} of Algorithm \ref{alg:quanrud_cardinality}, and our fast span computation algorithm presented in Section \ref{thm:span_computation} to implement Line~\ref{linec:span_computation}.
See Section \ref{subsec:our_cardinality_linear_mi} for details.

\subsection{Maximum Weight Matroid Intersection}

In this subsection, we describe Quanrud's framework for maximum weight matroid intersection.
The high-level approach in the weighted case is similar to that in the unweighted case, but requires additional technical sophistication.
The pseudocode is shown in Algorithm \ref{alg:quanrud_weighted}.

\begin{algorithm}[t]
    Let $w(e) \gets 1$ for all $e \in V$. \\
    \For(\tcp*[f]{$L = O(\log (n) / \eps)$}){$\ell = 1$ to $L$ } {
        Sample a set $V' \subseteq V$ of $\Theta(\min\{ k \log(n) / \eps, n \})$ elements, where each element $e \in V$ is chosen with probability proportional to $w(e)$. \label{linew:sample} \\
        Compute a $(1 - \eps)$-approximate maximum-weight common independent set $I^{(\ell)}$ and a compact $(1 - \eps)$-approximate dual solution $\tilde{y}^{(\ell)}, \tilde{z}^{(\ell)} : 2^{V'} \to \mathbb{R}_{\geq 0}$ for the subproblem over $V'$. \label{linew:small_computation} \\
        Define $y^{(\ell)}, z^{(\ell)} : 2^{V} \to \mathbb{R}_{\geq 0}$ by 
        \begin{align*}
            y^{(\ell)}(\mathrm{span}_{\M_1}(S')) &= \lceil \tilde{y}^{(\ell)}(S') \rceil_{1 + \varepsilon} \quad \quad \text{for all } S' \in \mathrm{support}(\tilde{y}^{(\ell)}), \\
            z^{(\ell)}(\mathrm{span}_{\M_2}(T')) &= \lceil \tilde{z}^{(\ell)}(T') \rceil_{1 + \varepsilon} \quad \quad \text{for all } T' \in \mathrm{support}(\tilde{z}^{(\ell)}),
        \end{align*} 
        where $\lceil x \rceil_{1 + \varepsilon}$ = $(1 + \eps)^{\lceil \log_{1 + \eps} x \rceil}$ rounds $x$ up to the nearest power of $1 + \eps$. \label{linew:span_computation}\\
        For each element $e \in V$, if $\displaystyle \sum_{S \ni e} (y(S) + z(S)) \geq c(e)$ holds, then set $w(e) \gets w(e) / 2$. \label{linew:update_weights}
    }
    Output $I^{(\ell)}$ for some iteration $\ell \in [L]$ for which $c(I^{(\ell)})$ is maximized over all iterations.
    \caption{Quanrud's maximum weight matroid intersection algorithm in the independence oracle model (from {\cite[Figure 2]{quanrud:LIPIcs.ICALP.2024.118}})} \label{alg:quanrud_weighted}
\end{algorithm}


Here, we refer to a solution $(y, z)$, where $y, z : 2^V \to \mathbb{R}_{\geq 0}$, of the following LP as a dual solution.
This LP is the dual of the LP relaxation of the weighted matroid intersection problem.

\begin{equation} \label{equ:dual_lp}
\begin{aligned} 
    \mathrm{minimize} & \,\, \sum_{S \subseteq V} \left(\mathrm{rank}_{\M_1}(S) y(S) + \mathrm{rank}_{\M_2}(S) z(S) \right) \quad \mathrm{over} \,\,\, y, z : 2^V \to \mathbb{R}_{\geq 0} \\
    \mathrm{s.t.} & \,\, \sum_{S \ni e} (y(S) + z(S)) \geq c(e) \quad \mathrm{for} \,\, \mathrm{all} \,\, e \in V 
\end{aligned}
\end{equation}

Quanrud defines $(y, z)$ to be {\em compact} if the supports of $y$ and $z$ satisfy the following conditions:
\begin{align*}
    \mathrm{support}(y) \subseteq& \{ \mathrm{span}_{\M_1}(\{e_1\}), \mathrm{span}_{\M_1}(\{e_1, e_2\}), \ldots, \mathrm{span}_{\M_1}(\{e_1, \ldots, e_k\}) \}, \\
    \mathrm{support}(z) \subseteq& \{ \mathrm{span}_{\M_2}(\{f_1\}), \mathrm{span}_{\M_2}(\{f_1, f_2\}), \ldots, \mathrm{span}_{\M_2}(\{f_1, \ldots, f_k\}) \}.
\end{align*}
Here, $e_1, \ldots, e_k \in V$ and $f_1, \ldots, f_k \in V$ are two sequences of $k$ elements.

As in the unweighted case, we maintain a weight function $w : V \to \mathbb{R}_{\geq 0}$ with initially equal weights, and repeat the following process for $L = O(\log(n) / \eps)$ iterations:
In each iteration, we first sample $\tilde{O}(k / \eps)$ elements $V' \subseteq V$ in proportion to $w$.
Then, we compute a $(1 - \eps)$-approximate maximum-weight common independent set $I$ and a compact $(1 - \eps)$-approximate dual solution $\tilde{y}, \tilde{z} : 2^{V'} \to \mathbb{R}_{\geq 0}$ for the subproblem over $V'$.
We then construct $y, z : 2^V \to \mathbb{R}_{\geq 0}$ from $\tilde{y}, \tilde{z}$ by setting
\begin{align*}
    y(\mathrm{span}_{\M_1}(S')) &= \lceil \tilde{y}(S') \rceil_{1 + \varepsilon} \quad \quad \text{for all } S' \in \mathrm{support}(\tilde{y}), \\
    z(\mathrm{span}_{\M_2}(T')) &= \lceil \tilde{z}(T') \rceil_{1 + \varepsilon} \quad \quad \text{for all } T' \in \mathrm{support}(\tilde{z}).
\end{align*} 
We set $y(F) = 0$ (resp. $z(F) = 0$) for any $F \in 2^V$ not of the form $\mathrm{span}_{\M_1}(S')$ with $S' \in \mathrm{support}(\tilde{y})$ (resp. $\mathrm{span}_{\M_2}(T')$ with $T' \in \mathrm{support}(\tilde{z})$).
Here, $\lceil x \rceil_{1 + \varepsilon}$ = $(1 + \eps)^{\lceil \log_{1 + \eps} x \rceil}$ denote rounding $x$ up to the nearest power of $1 + \eps$.
Since $(y', z')$ is compact, there exists a sequence $e_1, \ldots, e_k$ such that every set in the support of $y'$ has the form $S'_i = (\mathrm{span}_{\M_1}(\{ e_1, \ldots , e_i \})) \cap V'$.
Then, every set in the support of $y$ has the form $S_i = \mathrm{span}_{\M_1}(S'_i) = \mathrm{span}_{\M_1}(\{ e_1, \ldots , e_i \})$.
Symmetrically, the same property holds for $z$, with a different sequence of $k$ elements determined by $z'$.
Thus, $(y, z)$ is also compact.
Finally, we update the weights $w$ by setting $w(e) \gets w(e) / 2$ for all elements $e \in V$ such that $\sum_{S \ni e} (y(S) + z(S)) \geq c(e)$ holds.
After $L$ iterations, we output $I$ that maximizes $c(I)$ over all iterations.

When applying Quanrud's framework to maximum weight linear matroid intersection, we use the $(1-\eps)$-approximate weighted linear matroid intersection algorithm by Huang--Kakimura--Kamiyama~\cite{huang2016exact} to implement Line~\ref{linew:small_computation} of Algorithm~\ref{alg:quanrud_weighted}.
To implement Line~\ref{linew:update_weights}, we combine our fast span computation algorithm presented in Section~\ref{thm:span_computation} with the fast Gram--Schmidt orthonormalization algorithm of van den Brand~\cite{van2021unifying} (Lemma~\ref{lem:vandenbrand_qr_decomp}).
See Section \ref{subsec:our_weight_linear_mi} for details.

\section{Linear Matroid Intersection Algorithms} \label{sec:proof_of_main_thm}

In this section, we prove Theorems \ref{thm:cardinality_main_thm} and \ref{thm:weighted_main_thm} by combining the sparsification framework of Quanrud \cite{quanrud:LIPIcs.ICALP.2024.118} (see Section~\ref{sec:quanrud} for details) with our fast span computation algorithm presented in Section \ref{thm:span_computation}.

When we apply Quanrud's framework to linear matroid intersection, we set $k = r$ instead of $k = r_*$.
In the proofs of Theorems \ref{thm:cardinality_main_thm} and \ref{thm:weighted_main_thm}, we apply Quanrud's framework to matrices that have been suitably preprocessed using the sketching technique of Cheung--Kwok--Lau~\cite{cheung2013fast} (Lemma~\ref{lem:sketching_tech}).
For the preprocessed matrices, the maximum size of a common independent set is within a constant factor of the number of rows, which justifies setting $k = r$.\footnote{
     Actually, sampling $\Theta(\min\{r \log (n) / \eps, n \})$ elements, instead of $\Theta(\min\{ r_* \log (n) / \eps, n \})$, in each iteration does not affect the correctness of Quanrud's framework.
     This is because the approximation guarantee relies on Lemma 9 (and Lemma 13 in the weighted case) in Quanrud's paper {\cite{quanrud:LIPIcs.ICALP.2024.118}}, which states that sampling $\Omega(\min\{ r_* \log (n) / \eps, n \})$ elements is sufficient.
}

\subsection{Maximum Cardinality Linear Matroid Intersection} \label{subsec:our_cardinality_linear_mi}

We first show the following theorem, and then prove Theorem \ref{thm:cardinality_main_thm} using this theorem.

\begin{theorem} \label{thm:cardinality_little_weak_thm}
    For any $\eps > 0$, a $(1 - \eps)$-approximate solution to the maximum cardinality linear matroid intersection problem can be computed with high probability in $\tilde{O}_{\eps}(\mathrm{nnz}(M_1) + \mathrm{nnz}(M_2) + r^{\omega})$ time.
\end{theorem}

\begin{proof}
We show that, when the matroids are given by matrix representations, Quanrud's maximum cardinality matroid intersection algorithm (Algorithm \ref{alg:quanrud_cardinality}) can be implemented in $\tilde{O}_{\eps}(\mathrm{nnz}(M_1) + \mathrm{nnz}(M_2) + r^{\omega})$ time.
By using Harvey's exact algorithm for linear matroid intersection (Theorem \ref{thm:Harvey}), Line \ref{linec:small_computation} can be implemented in $\tilde{O}_{\eps}(r^{\omega})$ time, since the number of elements in the sparsified set is $\tilde{O}_{\eps}(r)$.
By applying Lemma \ref{lem:span_computation}, Line \ref{linec:span_computation} can be implemented in $\tilde{O}(\mathrm{nnz}(M_1) + \mathrm{nnz}(M_2) + r^{\omega})$ time.
Furthermore, Lines \ref{linec:sample} and \ref{linec:update_weights} can be implemented in $\tilde{O}_{\eps}(n)$ time. 
Since the number of iterations of the for loop is $O(\log(n) / \eps)$, the total running time is the desired bound of $\tilde{O}_{\eps}(\mathrm{nnz}(M_1) + \mathrm{nnz}(M_2) + r^{\omega})$, which completes the proof. 
\end{proof}

Following the approach of \cite[Section 4.3.2]{cheung2013fast}, we can improve the running time to $\tilde{O}_{\eps}(\mathrm{nnz}(M_1) + \mathrm{nnz}(M_2) + r_*^{\omega})$, which proves Theorem~\ref{thm:cardinality_main_thm}.

\begin{proof}[Proof of Theorem {\ref{thm:cardinality_main_thm}}]
We assume without loss of generality that $\eps < 1 / 2$.

Given an integer $\ell$, we consider the task of finding a set $I' \subseteq [n]$ of size at least $\min\{ (1 - \varepsilon)\ell, (1 - \varepsilon) \cdot r_* \}$ such that the columns indexed by $I'$ are linearly independent in both $M_1$ and $M_2$.
Using Lemma~\ref{lem:sketching_tech} with $k = \ell$, we compress the input matrices $M_1$ and $M_2$ into $O(\ell) \times n$ matrices $M_1'$ and $M_2'$, in $O(\mathrm{nnz}(M_1))$ and $O(\mathrm{nnz}(M_2))$ time, respectively.
We have $\mathrm{nnz}(M_1') = O(\mathrm{nnz}(M_1))$ and $\mathrm{nnz}(M_2') = O(\mathrm{nnz}(M_2))$.
Furthermore, for any set $S \subseteq [n]$ of size at most $\ell$, if the columns indexed by $S$ are linearly independent in both $M_1$ and $M_2$, then the columns indexed by $S$ are also linearly independent in both $M_1'$ and $M_2'$ with probability at least $1 - O(1/n^{1/3})$.
Consequently, the maximum size of a common independent set for $M'_1$ and $M'_2$ is $\min\{ \ell, r_* \}$ with high probability.
Then, we apply the algorithm in Theorem \ref{thm:cardinality_little_weak_thm} to find a $(1 - \eps)$-approximate solution to the linear matroid intersection problem on $M_1'$ and $M_2'$, which runs in $\tilde{O}_{\eps}(\mathrm{nnz}(M_1') + \mathrm{nnz}(M_2') + \ell^{\omega}) = \tilde{O}_{\eps}(\mathrm{nnz}(M_1) + \mathrm{nnz}(M_2) + \ell^{\omega})$ time.

To find a solution of size at least $(1 - \varepsilon) \cdot r_*$, we set $\ell = 2, 4, 8, \ldots$ and apply the above algorithm until there is no solution of size $\ell / 2$.
The total running time of this algorithm is $\tilde{O}_{\eps}(\mathrm{nnz}(M_1) + \mathrm{nnz}(M_2) + r_*^{\omega})$, which completes the proof. 
\end{proof}

\subsection{Maximum Weight Linear Matroid Intersection} \label{subsec:our_weight_linear_mi}

We first show the following theorem, and then prove Theorem \ref{thm:weighted_main_thm} using this theorem.

\begin{theorem}  \label{thm:weighted_main_week_thm}
    For any $\eps > 0$, a $(1 - \eps)$-approximate solution to the maximum weight linear matroid intersection problem can be computed with high probability in $\tilde{O}_{\eps}(\mathrm{nnz}(M_1) + \mathrm{nnz}(M_2) + r^{\omega})$ time.
\end{theorem}

In Quanrud's framework, we need to compute a $(1 - \eps)$-approximate maximum-weight common independent set and a compact $(1 + \eps)$-approximate dual solution for the sparsified instances.
When implementing Quanrud's framework for linear matroids, we use the following lemma to compute compact primal-dual solutions.
This lemma is obtained by combining the $(1- \eps)$-approximation algorithm for the weighted matroid intersection by Huang--Kakimura--Kamiyama~\cite{huang2016exact} with the argument of Quanrud's paper~\cite{quanrud:LIPIcs.ICALP.2024.118}.
See Appendix~\ref{appendix:weighted_sparsified_instance} for the proof of Lemma~\ref{lem:weighted_sparsified_instance}.

\begin{lemma} \label{lem:weighted_sparsified_instance}
    Given two $r \times n$ matrix $M_1$ and $M_2$ over a field $\F$, and a weight function $c : [n] \to \mathbb{R}_{\geq 0}$, there is an algorithm that computes a $(1 - \eps)$-approximate maximum-weight common independent set $I$, along with a compact $(1 + \eps)$-approximate dual solution $y, z : 2^{[n]} \to \R_{\geq 0}$, with high probability in $\TO(nr^{\omega - 1})$ time.
\end{lemma}

To efficiently implement Line~\ref{linew:update_weights} of Algorithm \ref{alg:quanrud_weighted}, we present the following lemma.
This lemma is obtained by combining our fast span computation algorithm described in Section~\ref{thm:span_computation} with the fast Gram--Schmidt orthonormalization algorithm of van den Brand (Lemma~\ref{lem:vandenbrand_qr_decomp}).

\begin{lemma}  \label{lem:weighted_dynamic_fast_span_computation}
    Let $M$ be an $r \times n$ matrix over a field $\F$, and let $\{ i_1, \ldots, i_{r} \} \subseteq [n]$ be a set of indices.
    Here, $\boldsymbol{m}_1, \ldots, \boldsymbol{m}_n$ denote the vectors corresponding to the columns of $M$.
    Then, there is an algorithm that answers the following query: given $j \in [n]$ and $\alpha \in [r]$, determine whether $\boldsymbol{m}_j \in \mathrm{span} \{ \boldsymbol{m}_{i_\beta} \mid \beta \leq \alpha \}$ with high probability in $O(\nnz(\boldsymbol{m}_j))$ time.
    This algorithm requires $O(r^\omega)$ time for preprocessing.
\end{lemma}

Here, we provide a proof only for matrices over $\mathbb{R}$.
The extension to matrices over an arbitrary field is deferred to Section~\ref{appendix:general_orthogonalization}
To make the algorithm applicable over arbitrary fields, the key modification is to replace the Euclidean inner product $\bv^\top \bw$ with the bilinear form $\langle \boldsymbol{v}, \boldsymbol{w} \rangle = \bv^\top B \bw$ (see Section~\ref{appendix:fast_span_any_field} for details).

\begin{proof}[Proof of Lemma~\ref{lem:weighted_dynamic_fast_span_computation} for matrices over $\mathbb{R}$]
We first describe the preprocessing algorithm.
First, we apply the fast Gram--Schmidt orthonormalization algorithm of van den Brand~\cite{van2021unifying} (Lemma~\ref{lem:vandenbrand_qr_decomp}) to the sequence of $2 r$ vectors $\bm_{i_1}, \ldots, \bm_{i_r}, \be_1, \ldots, \be_r$, which takes $O(r^\omega)$ time.
Here, $\boldsymbol{e}_i \in \R^r$ is the standard basis vector (i.e., the vector with $1$ in the $i$-th coordinate and $0$ elsewhere).
Then, we obtain a vector sequence $\bm'_{i_1}, \ldots, \bm'_{i_r}, \be'_1, \ldots, \be'_r$.
We note that this vector sequence includes the zero vector $\boldsymbol{0}$.
We then choose $q_1, \ldots, q_{2r}$ independently and uniformly at random from a finite subset $R \subset \R$ of size $\mathrm{poly}(n)$ (e.g., $R = \{ 1, 2, \ldots, \text{poly}(n) \}$).
Next, for each $\alpha \in [r]$, we compute a vector $\displaystyle \bv_\alpha = \sum_{j = \alpha + 1}^{r} q_j \boldsymbol{m}'_{i_j} + \sum_{j = 1}^r q_{r + j}\boldsymbol{e}'_j$.
This can be implemented in $O(r^2)$ time by computing the sums in reverse order.
We note that $\bv_\alpha$ lies in the orthogonal complement of $\mathrm{span}\{ \bm_{i_1}, \ldots, \bm_{i_\alpha} \}$.

Next, we describe how to answer the query.
It suffices to check whether $\bv_\alpha^\top \bm_j = 0$, which can be done in $O(\nnz(\boldsymbol{m}_j))$ time.
As in the proof of Lemma~\ref{lem:span_computation}, if $\boldsymbol{m}_j \in \mathrm{span} \{ \boldsymbol{m}_{i_\beta} \mid \beta \leq \alpha \}$ holds, then $\bv_\alpha^\top \bm_j = 0$ always holds, otherwise $\bv_\alpha^\top \bm_j \neq 0$ holds with probability at least $1 - 1/|R| = 1 - 1 / \text{poly}(n)$.
Therefore, checking whether $\boldsymbol{v}_\alpha^{\top}\boldsymbol{m}_j=0$ suffices to decide whether $\boldsymbol{m}_j \in \mathrm{span} \{ \boldsymbol{m}_{i_\beta} \mid \beta \leq \alpha \}$, which completes the proof. 
\end{proof}

We now present the proof of Theorem~\ref{thm:weighted_main_week_thm}.

\begin{proof}[Proof of Theorem~\ref{thm:weighted_main_week_thm}]
We show that, when the matroids are given by matrix representations, Quanrud's maximum weight matroid intersection algorithm (Algorithm \ref{alg:quanrud_weighted}) can be implemented in $\tilde{O}_{\eps}(\mathrm{nnz}(M_1) + \mathrm{nnz}(M_2) + r^{\omega})$ time.
By Lemma~\ref{lem:weighted_sparsified_instance}, Line~\ref{linew:small_computation} can be implemented in $\TO(r^\omega)$ time, since the number of elements in the sparsified set is $\tilde{O}_{\eps}(r)$.

In our algorithm, we do not explicitly compute $(y, z)$ computed in Line~\ref{linew:span_computation}.
Here, $(y, z)$ computed in Line~\ref{linew:span_computation} is compact, and thus we can write
\begin{equation} \label{eq:support_eq}
\begin{aligned} 
    \mathrm{support}(y) \subseteq& \{ \mathrm{span}_{\M_1}(\{i_1\}), \mathrm{span}_{\M_1}(\{i_1, i_2\}), \ldots, \mathrm{span}_{\M_1}(\{i_1, \ldots, i_r\}) \}, \\
    \mathrm{support}(z) \subseteq& \{ \mathrm{span}_{\M_2}(\{j_1\}), \mathrm{span}_{\M_2}(\{j_1, j_2\}), \ldots, \mathrm{span}_{\M_2}(\{j_1, \ldots, j_r\}) \},
\end{aligned}
\end{equation}
for two sequences of $r$ elements $i_1, \ldots, i_r \in [n]$ and $j_1, \ldots, j_r \in [n]$.
Here, let $\M_1$ (resp. $\M_2$) denote the matroid represented by the matrix $M_1$ (resp. $M_2$).
The elements $i_1, \ldots, i_r$ and $j_1, \ldots, j_r$ belong to the sparsified set, and are computed in Line~\ref{linew:small_computation}.
To implement Line~\ref{linew:span_computation}, we simply compute the value of $y(\mathrm{span}_{\M_1}(\{ i_1, \ldots, i_\alpha \}))$ and $z(\mathrm{span}_{\M_2}(\{ j_1, \ldots, j_\alpha \}))$ for each $\alpha \in [r]$, which takes $O(r)$ time.
Note that we do not explicitly compute $\mathrm{span}_{\M_1}(\{ i_1, \ldots, i_\alpha \})$ or $\mathrm{span}_{\M_2}(\{ j_1, \ldots, j_\alpha \})$.

It remains to bound the running time of Line~\ref{linew:update_weights}.
In this line, for each $i \in [n]$, we need to check whether $\displaystyle \sum_{S \ni i} (y(S) + z(S)) \geq c(i)$.
Here, $\boldsymbol{m}^{(1)}_1, \ldots, \boldsymbol{m}^{(1)}_n$ (resp. $\boldsymbol{m}^{(2)}_1, \ldots, \boldsymbol{m}^{(2)}_n$) denote the vectors corresponding to the columns of $M_1$ (resp. $M_2$).
As preprocessing, for all $\alpha \in [r]$, we compute 
\begin{align*}
y_\alpha = \sum_{\gamma \in \{ \alpha, \alpha + 1, \ldots, r \}} y(\mathrm{span}_{\M_1}(\{ i_1, \ldots, i_\gamma \})),  \quad
z_\alpha = \sum_{\gamma \in \{ \alpha, \alpha + 1, \ldots, r \}} z(\mathrm{span}_{\M_2}(\{ j_1, \ldots, j_\gamma \})),
\end{align*}
which takes $O(r)$ time.
Furthermore, we apply the preprocessing algorithm of Lemma~\ref{lem:weighted_dynamic_fast_span_computation} to $M_1$ and $\{i_1, \ldots, i_r \}$, which takes $O(r^\omega)$ time; we do the same for $M_2$ and $\{ j_1, \ldots, j_r \}$.

We now describe how to check whether $\displaystyle \sum_{S \ni i} (y(S) + z(S)) \geq c(i)$ holds for each $i \in [n]$ in $\tO(\nnz(\bm^{(1)}_i) + \nnz(\bm^{(2)}_i))$ time.
We can binary search for the first index $\alpha^{(1)} \in [r]$ (resp. $\alpha^{(2)} \in [r]$) such that $\boldsymbol{m}^{(1)}_i \in \mathrm{span} \{ \boldsymbol{m}^{(1)}_{i_\beta} \mid \beta \leq \alpha^{(1)} \}$ (resp. $\boldsymbol{m}^{(2)}_i \in \mathrm{span} \{ \boldsymbol{m}^{(2)}_{j_\beta} \mid \beta \leq \alpha^{(2)} \}$).
Note that, for all $\gamma \geq \alpha^{(1)}$, we have $\bm^{(1)}_i \in \mathrm{span} \{ \boldsymbol{m}^{(1)}_{i_\beta} \mid \beta \leq \gamma \}$.
By Lemma~\ref{lem:weighted_dynamic_fast_span_computation}, this binary search process takes $\tO(\nnz(\bm^{(1)}_i) + \nnz(\bm^{(2)}_i))$ time.
Here, we have $\displaystyle \sum_{S \ni i} (y(S) + z(S)) = y_{\alpha^{(1)}} + z_{\alpha^{(2)}}$, so this value can be computed in $O(1)$ time.
Thus, for each $i \in [n]$, we can check whether $\displaystyle \sum_{S \ni i} \left(y(S) + z(S)\right) \geq c(i)$ in $\tO(\nnz(\bm^{(1)}_i) + \nnz(\bm^{(2)}_i))$ time.
Therefore, Line~\ref{linew:update_weights} can be implemented in $\tO(\nnz(M_1) + \nnz(M_2) + r^\omega)$ time, which completes the proof.
\end{proof}

Using the sketching technique of Cheung--Kwok--Lau~\cite{cheung2013fast} (Lemma~\ref{lem:sketching_tech}), we can improve the running time to $\tilde{O}_{\eps}(\mathrm{nnz}(M_1) + \mathrm{nnz}(M_2) + r_*^{\omega})$, which proves Theorem~\ref{thm:weighted_main_thm}.

\begin{proof}[Proof of Theorem \ref{thm:weighted_main_thm}]
    As a first step, we compute a $2/3$-approximate solution $\bar{S_*}$ to the maximum cardinality linear matroid intersection for the input matrices $M_1$ and $M_2$, using Theorem~\ref{thm:cardinality_main_thm}, in $\tilde{O}_{\eps}(\mathrm{nnz}(M_1) + \mathrm{nnz}(M_2) + r_{*}^{\omega})$ time.
    We then apply Lemma \ref{lem:sketching_tech} with $k = \frac{3}{2}|\bar{S_*}|$.
    Note that $k \geq r_{*}$ and $k = \Theta(r_*)$.
    The lemma compresses the input matrices $M_1$ and $M_2$ into $O(r_*) \times n$ matrices $M_1'$ and $M_2'$, in $O(\mathrm{nnz}(M_1))$ and $O(\mathrm{nnz}(M_2))$ time, respectively.
    We have $\mathrm{nnz}(M_1') = O(\mathrm{nnz}(M_1))$ and $\mathrm{nnz}(M_2') = O(\mathrm{nnz}(M_2))$.
    Furthermore, for any set $S \subseteq [n]$ of size at most $r_*$, if the columns indexed by $S$ are linearly independent in both $M_1$ and $M_2$, then the columns indexed by $S$ are also linearly independent in both $M_1'$ and $M_2'$ with probability at least $1 - O(1/n^{1/3})$.
    We then apply the algorithm in Theorem~\ref{thm:weighted_main_week_thm} to the compressed matrices $M_1'$ and $M_2'$ to find a $(1 - \varepsilon)$-approximate solution to the maximum weight linear matroid intersection problem.
    This algorithm runs in $\tilde{O}_{\varepsilon}(\mathrm{nnz}(M_1') + \mathrm{nnz}(M_2') + r_*^{\omega}) = \tilde{O}_{\varepsilon}(\mathrm{nnz}(M_1) + \mathrm{nnz}(M_2) + r_*^{\omega})$ time. 
\end{proof}


\section{Extending the Algorithm to Arbitrary Fields} \label{appendix:fast_span_any_field}



In this section, we develop the tools required to prove Lemma~\ref{lem:weighted_dynamic_fast_span_computation} over an arbitrary field.

Note that we assume $|\mathbb{F}| = \Omega(\mathrm{poly}(n))$.
Recall that, by Lemma \ref{lem:matrix_elem_assumption}, this assumption can be made without loss of generality.
We also assume that $r \leq O(n)$, which can be assumed without loss of generality, since, even when $r > n$, we can reduce the dimension of $M$ from $r \times n$ to $O(n) \times n$ in $O(\mathrm{nnz}(M))$ time by applying Lemma~\ref{lem:sketching_tech} with $k = n$, while preserving the linear independence of the column set.

To extend the algorithm to arbitrary fields, we replace the Euclidean inner product $\bv^\top \bw$ with the bilinear form $\langle \boldsymbol{v}, \boldsymbol{w} \rangle = \bv^\top B \bw$, where we choose $w_1, \ldots, w_r$ independently and uniformly at random from a finite subset $F \subseteq \F$ of size $\text{poly}(n)$, and define a diagonal $r \times r$ matrix $B$ such that $B_{i, i} = w_i$.

We define the orthogonal complement with respect to the bilinear form $\langle \boldsymbol{v}, \boldsymbol{w} \rangle = \bv^\top B \bw$ as $\Wp_B = \{ \boldsymbol{v} \in \mathbb{F}^r \mid \forall \boldsymbol{w} \in W, \langle \boldsymbol{v}, \boldsymbol{w} \rangle = \bv^\top B \bw = 0  \}$.
For a fixed linear subspace $W$, we can show that $W \cap \Wp_B = \{ \boldsymbol{0} \}$ and $W \oplus \Wp_B = \mathbb{F}^r$ with high probability.
This property is important for ensuring that our algorithm works over an arbitrary field.
To show this property, we first prove the following claim.

\begin{claim} \label{claim:imitate_matroid_intersection}
    Let $\ell \leq r$, and let $B$ be a diagonal $r \times r$ random matrix such that $B_{i, i} = w_i$.
    For any $r \times \ell$ matrix $M$ of rank $\ell$, we have $\mathrm{rank}(M^\top B M) = \ell$ with high probability over the choice of $w_1, \ldots, w_r$ chosen independently and uniformly at random from a finite subset $F \subseteq \F$ of size $\text{poly}(n)$.
\end{claim}

To prove the claim, we use an idea from the matrix formulation of linear matroid intersection (see e.g., \cite[Lemma 2.5]{gurjar2020linear}), originally noted by Tomizawa--Iri~\cite{tomizawa1974algorithm}.

\begin{proof}
Here, we define a symbolic diagonal $r \times r$ matrix $D$ such that $D_{i, i} = x_i$.
By the Cauchy--Binet formula\footnote{Let $A_1, A_2$ be $r \times n$ matrices. Then, $\det (A_1A_2^\top) = \sum_{S \in \binom{[n]}{r}} \det A_1[S]\det A_2[S]$.}, 
\begin{align*}
    \det(M^\top D M) = \sum_{S \in \binom{[r]}{\ell}} \det M^\top[S] \det(DM)^\top[S] = \sum_{S \in \binom{[r]}{\ell}} \det M^\top[S] \det M^\top [S] \prod_{i \in S} x_i .
\end{align*}
Since $\mathrm{rank}(M) = \ell$, there exists some $S \in \binom{[r]}{\ell}$ such that $\det M^\top[S] \neq 0$.
Thus, $R(x_1, \ldots, x_r) = \det(M^\top D M)$ is a nonzero polynomial of total degree $\ell$.
Because $w_1, \ldots, w_r$ are chosen independently and uniformly at random from $\F$, the Schwartz--Zippel Lemma (Lemma \ref{lem:schwartz_zippel}) implies that the probability that $\det(M^\top B M) = R(w_1, \ldots, w_r) = 0$ is at most $\ell/|\F|$. 
Here, we recall that $|\F| = \mathrm{poly}(n)$ by assumption.
Therefore, $\mathrm{rank}(M^\top B M) = \ell$ holds with high probability, which completes the proof. 
\end{proof}

Using the claim, we prove the following lemma.

\begin{lemma} \label{lem:orthogonal_complement_good_property}
    Fix any linear subspace $W \subseteq \F^r$.
    With high probability, we have $W \cap \Wp_B = \{ \boldsymbol{0} \}$, $W \oplus \Wp_B = \mathbb{F}^r$, and $(\Wp_B)^\perp_B = W$.
\end{lemma}

\begin{proof}
We first prove that $W \cap \Wp_B = \{ \boldsymbol{0} \}$ holds with high probability.
To this end, we consider the properties of a column vector $\boldsymbol{x} \in W \cap \Wp_B$.
For any column vector $\by \in W$, we have $ \langle \by, \bx \rangle = \by^\top B \bx = 0$.
Let $M$ be an $r \times \ell$ matrix whose columns are the basis vectors $\bm_1, \ldots, \bm_\ell$ of $W$.
Note that $W = \mathrm{span}\{ \bm_1, \ldots, \bm_{\ell} \}$ and $\mathrm{rank}(M) = \ell$.
Since $\bx \in W$, we can write $\bx = M \bu$ for some $\bu \in F^\ell$.
Then, for any $\bv \in \F^\ell$, we have
\begin{equation*}
(M \bv)^\top B M \bu = \bv^\top M^\top B M \bu = 0.
\end{equation*}
Now, let $A = M^\top B M$.
Then, for any $v \in \F^\ell$, we have $(A^\top \bv)^\top \bu = 0$.
By the claim, the $\ell \times \ell$ matrix $A^\top$ has full rank with high probability.
Thus, for any $\bw \in \F^\ell$, $\bw^\top \bu = 0$ holds. 
In particular, by choosing $\bw$ to be the standard basis vector $\boldsymbol{e}_i$ (i.e., the vector with $1$ in the $i$-th coordinate and $0$ elsewhere), we obtain that the $i$-th coordinate of $\bu$ is $0$.
Since this holds for all $i \in [\ell]$, we conclude that $\bu = \boldsymbol{0}$.
Thus, $\bx = \boldsymbol{0}$, and hence $W \cap \Wp_B = \{ \boldsymbol{0} \}$ holds with high probability.

Next, we prove that $W \oplus \Wp_B = \mathbb{F}^r$ holds with high probability.
Now, $\bx \in \Wp_B$ if and only if $\bm_i^\top B \bx = 0$ for any $i \in [\ell]$, where we recall that $\bm_1, \ldots, \bm_\ell$ form a basis of $W$.
Thus, we have $\Wp_B = \{ \bx \in \F^r \mid M^\top B \bx = \boldsymbol{0} \}$.
Here, by the claim, $\mathrm{rank}(M^\top B M) = \ell$ holds with high probability.
Thus, we have $\mathrm{rank}(M^\top B) = \ell$, where we recall that $M^\top B$ is an $\ell \times r$ matrix.
Then, by the Rank-nullity theorem, $\mathrm{dim}(\Wp_B) = r - \ell$.
Since $\dim(W) = \ell$ and $W \cap \Wp_B = \{ \boldsymbol{0} \}$, $W \oplus \Wp_B = \mathbb{F}^r$ holds with high probability.

It remains to prove that $(\Wp_B)^\perp_B = W$ holds with high probability.
Now, fix any $\bx \in W$.
Here, for any $\by \in \Wp_B$, we have $\langle \by, \bx \rangle = 0$.
Thus, $\bx \in (\Wp_B)^\perp_B$, and therefore, $W \subseteq (\Wp_B)^\perp_B$.
Here, since $W \oplus \Wp_B = \mathbb{F}^r$, we can write $\Wp_B = \mathrm{span}\{ \bb_1, \ldots, \bb_{r - \ell} \}$, where the vectors $\bb_1, \ldots, \bb_{r - \ell}$ are independent of random variables $w_1, \ldots, w_r$.
Then, by applying the same argument as in the proof of $W \oplus \Wp_B = \mathbb{F}^r$, we have $\Wp_B \oplus (\Wp_B)^\perp_B = \F^r$ with high probability.
Thus, $\dim((\Wp_B)^\perp_B) = \ell$, and therefore $(\Wp_B)^\perp_B = W$ holds with high probability, which completes the proof. 
\end{proof}

In the fast span computation algorithm presented in Section~\ref{thm:span_computation}, we sample a vector $\bv$ from the orthogonal complement of $\mathrm{span}\{ \bm_j \mid j \in S \}$ with respect to the Euclidean inner product $\bv^\top \bw$, and then check whether $\bv^\top \bm_i = 0$ for all $i \in [n]$.
We showed that, to determine whether $\boldsymbol{m}_i \in \mathrm{span} \{ \boldsymbol{m}_j \mid j \in S \}$, it suffices to check whether $\bv^\top \bm_i = 0$.
The following lemma shows that that the same property holds even when the vector is sampled from the orthogonal complement with respect to the bilinear form $\langle \boldsymbol{v}, \boldsymbol{w} \rangle = \bv^\top B \bw$.

\begin{lemma} \label{lem:span_property_honshitsu}
    Let $M$ be an $r \times n$ matrix over a field $\mathbb{F}$ whose columns represent the elements of a matroid $\mathcal{M}$, and let $\boldsymbol{m}_1, \ldots, \boldsymbol{m}_n$ denote the column vectors of $M$.
    Let $S \subseteq [n]$.
    Denote by $W = \mathrm{span}\{ \bm_j \mid j \in S \}$, and let $\{ \boldsymbol{b}_1, \ldots, \boldsymbol{b}_{k} \}$ form a basis of $W^\perp_B$.
    Define $\displaystyle \boldsymbol{v} = \sum_{j = 1}^{k} r_j \boldsymbol{b}_j$, where $r_1, \ldots, r_k$ are chosen independently and uniformly at random from a finite subset $F \subseteq \F$.  
    Then, the value of $\langle \boldsymbol{v}, \boldsymbol{m}_i \rangle = \bv^\top B \bm_i$ satisfies the following two properties.
    \begin{itemize}
        \item If $i \in \mathrm{span}_{\M}(S)$, then we have $\langle \boldsymbol{v}, \boldsymbol{m}_i \rangle = \bv^\top B \bm_i = 0$
        \item If $i \notin \mathrm{span}_{\M}(S)$, then we have $\langle \boldsymbol{v}, \boldsymbol{m}_i \rangle = \bv^\top B \bm_i \neq 0$ with probability at least $1 - 1/|F|$.
    \end{itemize}
\end{lemma}

\begin{proof}


Similar to the proof of Lemma~\ref{lem:span_computation}, we can show that the value of $\langle \boldsymbol{v}, \boldsymbol{m}_i \rangle = \bv^\top B \bm_i$ satisfies the following two properties.

\begin{itemize}
    \item If $i \in \mathrm{span}_{\M}(S)$, then we have $\langle \boldsymbol{v}, \boldsymbol{m}_i \rangle = \bv^\top B \bm_i = 0$ for every possible choice of $r_1, \ldots, r_k$, since $\bm_i$ lies in $W = \mathrm{span}\{ \bm_i \mid i \in S \}$ and $\boldsymbol{v}$ lies in $\Wp_B = \{ \boldsymbol{x} \in \mathbb{F}^r \mid \forall \boldsymbol{y} \in W, \langle \boldsymbol{x}, \boldsymbol{y} \rangle = \bx^\top B \by = 0  \}$.
    
    \item If $i \notin \mathrm{span}_{\M}(S)$, then we have $\langle \boldsymbol{v}, \boldsymbol{m}_i \rangle = \bv^\top B \bm_i \neq 0$ with probability at least $1 - 1/|F|$. This can be shown as follows. Consider the multivariate polynomial 
    \begin{equation*}
        P(x_1, \ldots, x_k) = \sum_{j = 1}^{k} x_j \boldsymbol{b}_j^{\top} B \boldsymbol{m}_i.
    \end{equation*}
    Here, to show that $P(x_1, \ldots, x_k)$ is a nonzero polynomial, we prove the following claim.
    Recall that we are assuming $i \notin \mathrm{span}_{\M}(S)$.

    \begin{claim} \label{claim:helping_claim}
        If $\bm_i \notin W = \mathrm{span}\{ \bm_\alpha \mid \alpha \in S \}$, then, with high probability, there exists some $j \in [k]$ such that $\langle \bb_j, \bm_i \rangle = \bb_j^\top B \bm_i \neq 0$.
    \end{claim}
    \begin{proof}[Proof of Claim~\ref{claim:helping_claim}]
        Here, we have $\boldsymbol{m}_i \notin W$.
        Suppose, for the sake of contradiction, that $\langle \bb_j, \bm_i \rangle = \bb_j^\top B \bm_i = 0$ for all $j \in [k]$, where we recall that $\{ \boldsymbol{b}_1, \ldots, \boldsymbol{b}_k \}$ form a basis of $W^{\perp}_B$.
        This would imply that $\boldsymbol{m}_i \in (W^{\perp}_B)^{\perp}_B = W$, contradicting the assumption that $\boldsymbol{m}_i \notin W$.
        Here, we use the property $(W^{\perp}_B)^{\perp}_B = W$, which was shown in Lemma~\ref{lem:orthogonal_complement_good_property}.
        Therefore, there exists some $j \in [k]$ such that $\langle \bb_j, \bm_i \rangle = \bb_j^\top B \bm_i \neq 0$, which completes the proof.
    \end{proof}

    By Claim~\ref{claim:helping_claim}, $P(x_1, \ldots, x_k)$ is a nonzero polynomial of total degree $1$.
    Since $r_1, \ldots, r_k$ are chosen independently and uniformly at random from $F$, the Schwartz--Zippel Lemma (Lemma \ref{lem:schwartz_zippel}) implies that the probability that
    \begin{equation*}
        P(r_1, \ldots, r_k) = \sum_{j = 1}^{k} r_j \boldsymbol{b}_j^{\top} B \boldsymbol{m}_i = \boldsymbol{v}^{\top} B \boldsymbol{m}_i = \langle \bv, \bm_i \rangle = 0
    \end{equation*}
    is at most $1 / |F|$.
    Therefore, $\langle \boldsymbol{v}, \boldsymbol{m}_i \rangle = \bv^\top B \bm_i \neq 0$ with probability at least $1 - 1/|F|$, which completes the proof. 
\end{itemize}

\end{proof}

\section{Fast Orthogonalization Algorithm over an Arbitrary Field} \label{appendix:general_orthogonalization}

In this section, we provide a fast orthogonalization algorithm with respect to the bilinear form $\langle \bv, \bw \rangle = \bv^\top B \bw$, and then present a proof of Lemma~\ref{lem:weighted_dynamic_fast_span_computation}.
This orthogonalization algorithm is a modified version of the fast Gram--Schmidt orthonormalization algorithm of van den Brand~\cite{van2021unifying}, which we describe in Section~\ref{subsec:fast_orthogonal_basis}.


We first provide the following lemma on the orthogonal projection with respect to the bilinear form $\langle \boldsymbol{v}, \boldsymbol{w} \rangle = \bv^\top B \bw$, which is analogous to the fact that the matrix $A(A^\top A)^{-1}A^\top$ is the orthogonal projection onto the image of $A$.

\begin{lemma} \label{lem:orthogonal_projection_bilinear_form}
    Let $A$ be an $r \times \ell$ matrix of rank $\ell$.
    Let $W = \mathrm{span}\{ \ba_1, \ldots, \ba_\ell \}$, where $\ba_1, \ldots, \ba_\ell$ denote the vectors corresponding to the columns of $A$.
    For any $\bv \in \F^r$, we define $\bp = A(A^\top B A)^{-1}A^\top B \bv$ and $\bq = \bv - \bp$.
    Then, with high probability, we have $\bp \in W$ and $\bq \in \Wp_B$.
\end{lemma}

\begin{proof}
By Claim~\ref{claim:imitate_matroid_intersection}, we have $\mathrm{rank}(A^\top B A) = \ell$ with high probability; hence the $\ell \times \ell$ matrix $A^\top B A$ is invertible.
By Lemma~\ref{lem:orthogonal_complement_good_property}, $W \oplus \Wp_B = \F^r$ with high probability.
Thus, we can write $\bv = \bp + \bq$ for some $\bp \in W$ and $\bq \in \Wp_B$.
It suffices to show that $\bp = A(A^\top B A)^{-1}A^\top B \bv$.
For each $i \in [\ell]$, $\langle \ba_i, \bq \rangle = \langle \ba_i, \bv - \bp \rangle = 0$.
Since $\bp \in W$, we can write $\bp = A \bx$ for some $\bx \in \F^\ell$.
Thus, for each $i \in [\ell]$, we have $\ba_i^\top B \bv = \ba_i^\top B A x$.
Hence, $A^\top B \bv = A^\top B A \bx$, and thus $\bx = (A^\top B A)^{-1}A^\top B \bv$.
Therefore, $\bp = A \bx = A(A^\top B A)^{-1}A^\top B \bv$, which completes the proof. 
\end{proof}

Here, we refer to $\bp = A(A^\top B A)^{-1}A^\top B \bv$, as defined in Lemma~\ref{lem:orthogonal_projection_bilinear_form}, as the orthogonal projection of $v$ onto $W$ with respect to the bilinear form $\langle \bv, \bw \rangle = \bv^\top B \bw$.

We say that a linear subspace $W$ is {\em good}\footnote{In linear algebra terminology, $W$ is said to be non-degenerate with respect to the bilinear form $\langle \boldsymbol{v}, \boldsymbol{w} \rangle = \bv^\top B \bw$.} if it satisfies the following conditions: $W \cap \Wp_B = \{ \boldsymbol{0} \}$, $W \oplus \Wp_B = \mathbb{F}^r$, $(\Wp_B)^\perp_B = W$, and the orthogonal projection onto $W$ with respect to the bilinear form $\langle \bv, \bw \rangle = \bv^\top B \bw$ exists for every $\bv \in \F^r$.
By Lemmas~\ref{lem:orthogonal_complement_good_property} and \ref{lem:orthogonal_projection_bilinear_form}, any linear subspace $W$ is good with high probability.

We present the following lemma on a fast orthogonalization algorithm with respect to the bilinear form $\langle \bv, \bw \rangle = \bv^\top B \bw$, which is a modified version of Lemma~\ref{lem:vandenbrand_qr_decomp}.

\begin{lemma}[Fast orthogonalization algorithm with respect to the bilinear form $\langle \bv, \bw \rangle = \bv^\top B \bw$] \label{lem:bilinear_form_orthogonalization}
    Let $M$ be an $r \times n$ matrix over an arbitrary field with $r \leq n$, where $\boldsymbol{m}_1, \ldots, \boldsymbol{m}_n$ are column vectors of $M$.
    Here, we assume that the linear subspace $\mathrm{span}\{ \boldsymbol{m}_1, \ldots, \boldsymbol{m}_{i - 1} \}$ is good for all $i \in [n]$.
    Then, there is an algorithm that repeatedly computes $\boldsymbol{m}'_i = \boldsymbol{m}_i - P_{i - 1} \boldsymbol{m}_i$, where $P_{i - 1} \boldsymbol{m}_i$ denotes the orthogonal projection of $\boldsymbol{m}_i$ onto $\mathrm{span}\{ \boldsymbol{m}_1, \ldots, \boldsymbol{m}_{i - 1} \}$ with respect to the bilinear form $\langle \bv, \bw \rangle = \bv^\top B \bw$.
    The total running time of this process is $O(n^\omega)$.
\end{lemma}

The proof of Lemma~\ref{lem:bilinear_form_orthogonalization} is almost the same as \cite[Section 4.5]{van2021unifying}.
The only difference lies in the use of the Euclidean inner product $\boldsymbol{v}^\top \boldsymbol{w}$ versus the bilinear form $\langle \boldsymbol{v}, \boldsymbol{w} \rangle = \boldsymbol{v}^\top B \boldsymbol{w}$.

\begin{proof}
Throughout the proof, when we simply say orthogonal projection, we mean orthogonal projection with respect to the bilinear form $\langle \bv, \bw \rangle = \bv^\top B \bw$.
As in \cite[Section 4.5]{van2021unifying}, we use the data structures for dynamic matrix formula by van den Brand (\cite[Corollary 4.2]{van2021unifying}).

We first present the following claim, where the only difference from \cite[Lemma 4.8]{van2021unifying} is that our matrix formula incorporates the matrix $B$.

\begin{claim}[A modified version of {\cite[Lemma 4.8]{van2021unifying}}] \label{claim:modifieddynamic}
    Let $S \subseteq [i - 1]$ be a set of indices of a maximal independent subset of the first $i - 1$ columns of $M$, i.e., $\mathrm{span}\{ \bm_j \mid j \in S \} = \mathrm{span} \{ \bm_1 \ldots, \bm_{i - 1} \}$ and $\mathrm{rank}(\bm_1, \ldots, \bm_{i - 1}) = |S|$.
    We define $D_S$ as the $n \times n$ diagonal matrix whose $j$-th diagonal entry is $1$ if $j \in S$, and $0$ otherwise.
    Then, we have
    \begin{equation*}
        P_{i - 1} = M D_S ( D_S M^\top B M D_S + (I - D_S) )^{-1} D_S M^\top B.
    \end{equation*}
\end{claim}
\begin{proof}[Proof of Claim~\ref{claim:modifieddynamic}]
By Lemma~\ref{lem:orthogonal_projection_bilinear_form}, $M I[\ast, S] ( I[S, \ast] M^\top B M I[\ast, S] )^{-1} I[S, \ast] M^\top B$ is the orthogonal projection onto the image of $M I[\ast, S] = M[\ast, S]$, i.e., $\mathrm{span}\{ \bm_1, \ldots, \bm_{i - 1} \}$. 
Let $A = D_S M^\top B M D_S + (I - D_S)$ and $\bar{S} = [n] \setminus S$.
Here, we have
\begin{equation*}
    A^{-1}[S, S] = (I[S, \ast] M^\top B M I[\ast, S])^{-1}, A^{-1}[\bar{S}, \bar{S}] = I[\bar{S}, \bar{S}].
\end{equation*}
This can be verified by assuming, without loss of generality, that $S = [k]$ for $k = |S|$, in which case we have 
\begin{equation*}
D_S M^\top B M D_S + (I - D_S) = 
\begin{bmatrix}
I[S, \ast] M^\top B M I[\ast, S] & \\
 & I[\bar{S}, \bar{S}]
\end{bmatrix}.
\end{equation*}
Thus, to invert the matrix, we just have to invert each block.
In summary,
\begin{equation*}
M D_S ( D_S M^\top B M D_S + (I - D_S) )^{-1} D_S M^\top B = M I[\ast, S] ( I[S, \ast] M^\top B M I[\ast, S] )^{-1} I[S, \ast] M^\top B
\end{equation*}
is the orthogonal projection onto $\mathrm{span}\{ \bm_1, \ldots, \bm_{i - 1} \}$, which completes the proof. 
\end{proof}
As in \cite[Section 4.5]{van2021unifying}, let $S$ be initially an empty set and maintain
\begin{equation*}
(I - M D_S ( D_S M^\top B M D_S + (I - D_S) )^{-1} D_S M^\top B) M.
\end{equation*}
In $i$-th iteration, we query $i$-th column of this matrix formula and store the result as $\bm'_i$.
If $\bm'_i \neq \boldsymbol{0}$, then we add $i$ to set $S$; otherwise we do not.
As in \cite[Section 4.5]{van2021unifying}, each update and query can be performed in $O(n^{\omega - 1})$ time using the data structure from \cite[Corollary 4.2]{van2021unifying}.
Thus, the total running time is $O(n^\omega)$, which completes the proof. 
\end{proof}

We now present a proof of Lemma~\ref{lem:weighted_dynamic_fast_span_computation}, which is almost the same as the proof for matrices over $\mathbb{R}$ given in Section~\ref{subsec:our_weight_linear_mi}.
Note that we make use of the fact that, for a fixed linear subspace $W$, it holds with high probability that $W \cap W_B^{\perp} = { \boldsymbol{0} }$ and $W \oplus W_B^{\perp} = \mathbb{F}^r$.
This fact was established in Section~\ref{appendix:fast_span_any_field}.

\begin{proof} [Proof of Lemma~\ref{lem:weighted_dynamic_fast_span_computation}]
We first describe the preprocessing algorithm.
First, we apply the fast orthogonalization algorithm in Lemma~\ref{lem:bilinear_form_orthogonalization} to the sequence of $2 r$ vectors $\bm_{i_1}, \ldots, \bm_{i_r}, \be_1, \ldots, \be_r$, which takes $O(r^\omega)$ time.
Here, $\boldsymbol{e}_i \in \F^r$ is the standard basis vector (i.e., the vector with $1$ in the $i$-th coordinate and $0$ elsewhere).
Then, we obtain a vector sequence $\bm'_{i_1}, \ldots, \bm'_{i_r}, \be'_1, \ldots, \be'_r$.
We note that this vector sequence includes the zero vector $\boldsymbol{0}$.
We then choose $q_1, \ldots, q_{2r}$ independently and uniformly at random from a finite subset $F \subseteq \F$ of size $\text{poly}(n)$.
Next, for all $\alpha \in [r]$, we compute a vector $\displaystyle \bv_\alpha = \sum_{j = \alpha + 1}^{r} q_j \boldsymbol{m}'_{i_j} + \sum_{j = 1}^r q_{r + j}\boldsymbol{e}'_j$.
This computation can be implemented in $O(r^2)$ time by computing the sums in reverse order.
Let $W_\alpha = \mathrm{span}\{ \bm_{i_1}, \ldots, \bm_{i_\alpha} \}$.
By Lemma~\ref{lem:orthogonal_complement_good_property}, with high probability, $W_\alpha$ is good for all $\alpha \in [r]$.
Thus, $\bv_\alpha$ lies in $(W_\alpha)^\perp_B$.

Next, we describe how to answer the query.
It suffices to check whether $\langle \bv_\alpha, \bm_j \rangle = \bv_\alpha^\top B \bm_j = 0$, which takes $O(\nnz(\boldsymbol{m}_j))$ time.
By Lemma~\ref{lem:span_property_honshitsu}, if $\boldsymbol{m}_j \in \mathrm{span} \{ \boldsymbol{m}_{i_\beta} \mid \beta \leq \alpha \}$ holds, then $\langle \bv_\alpha, \bm_j \rangle = \bv_\alpha^\top B \bm_j = 0$ always holds, otherwise $\langle \bv_\alpha, \bm_j \rangle = \bv_\alpha^\top B \bm_j \neq 0$ holds with probability at least $1 - 1/|F| = 1 - 1 / \text{poly}(n)$.
Therefore, checking whether $\langle \bv_\alpha, \bm_j \rangle = \bv_\alpha^\top B \bm_j = 0$ suffices to decide whether $\boldsymbol{m}_j \in \mathrm{span} \{ \boldsymbol{m}_{i_\beta} \mid \beta \leq \alpha \}$, which completes the proof. 
\end{proof}


\bibliography{bib2doi}

\appendix

\section{Compact {$(1 \pm \varepsilon)$}-Primal-Dual Maximum-Weight Linear Matroid Intersection: Proof of Lemma~{\ref{lem:weighted_sparsified_instance}}} \label{appendix:weighted_sparsified_instance}

In this section, we prove Lemma~\ref{lem:weighted_sparsified_instance}.

To compute a $(1 - \eps)$-approximate primal solution and a compact $(1 + \eps)$-approximate dual solution, Quanrud~\cite{quanrud:LIPIcs.ICALP.2024.118} employs the $(1- \eps)$-approximation algorithm for weighted matroid intersection by Chekuri--Quanrud \cite{chekuri2016fast}, which computes a $(1 - \eps)$-approximate primal solution $I$ along with weight functions $c_1, c_2 : V \to \mathbb{R}_{\geq 0}$ that approximate Frank’s weight-splitting.
See properties $(a)$--$(f)$ on page 118:16 of Quanrud's \cite{quanrud:LIPIcs.ICALP.2024.118} paper for the properties that hold for the weight-splitting $c_1, c_2$ produced by the algorithm of Chekuri--Quanrud.
Using these weight functions $c_1$ and $c_2$, Quanrud constructs a compact $(1 + \eps)$-approximate dual solution.

Here, we use the $(1 - \eps)$-approximation algorithm for the weighted matroid intersection by Huang--Kakimura--Kamiyama~\cite{huang2016exact}.
Their algorithm also computes a $(1 - \eps)$-approximate primal solution $I$ along with weight functions $c_1$ and $c_2$ that approximate Frank’s weight-splitting.\footnote{
The key difference between Huang--Kakimura--Kamiyama~\cite{huang2016exact} and Chekuri--Quanrud~\cite{chekuri2016fast} is that the former uses an exact algorithm when solving unweighted matroid intersection as subproblem, whereas the latter uses a $(1-\eps)$-approximation algorithm.
In addition, by employing a more sophisticated weight adjustment, Chekuri--Quanrud achieve a faster algorithm in the independence oracle model than Huang--Kakimura--Kamiyama.
}
Note that the weight functions $c_1$ and $c_2$ computed by the algorithm of Huang--Kakimura--Kamiyama satisfy stronger properties than those computed by the algorithm of Chekuri--Quanrud.
Huang--Kakimura--Kamiyama described how to implement their algorithm when the matroids are given by matrix representations.

\begin{theorem}[$(1-\eps)$-approximate weighted linear matroid intersection algorithm by {\cite{huang2016exact}}] \label{thm:hkk_weighted}
    Given two matroids $\M_1 = (V, \I_1)$ and $\M_2 = (V, \I_2)$, both represented by $r \times n$ matrices, and a weight function $c : V \to \mathbb{R}_{\geq 0}$, there is an algorithm that, with high probability, runs in $\tilde{O}_{\eps}(n r^{\omega - 1})$ time and computes a $(1 - \eps)$-approximate maximum-weight common independent set $I$ along with weight functions $c_1, c_2 : V \to \mathbb{R}_{\geq 0}$ satisfying the following conditions:
    \begin{enumerate}[(a)]
        \item $c(e) \leq c_1(e) + c_2(e)$ for each element $e \in V$. 
        \item $c_1(I) + c_2(I) \leq (1 + \eps) c(I)$.
        \item $I$ is $c_1$-maximum independent set in $\M_1$.
        \item $I$ is $c_2$-maximum independent set in $\M_2$
    \end{enumerate}
\end{theorem}

\begin{remark} \label{remark:hkk_weighted}
    Theorem~\ref{thm:hkk_weighted} is not explicitly stated in the paper by Huang--Kakimura--Kamiyama~\cite{huang2016exact}.
    The theorem is based on a $(1 - \eps)$-approximation algorithm for the weighted matroid intersection algorithm, as given in \cite[Algorithm 2]{huang2016exact}, which employs Frank’s weight-splitting approach~\cite{frank1981weighted, frank2008quick}.
    See \cite[Section 5.3]{huang2016exact} for the implementation of their algorithm when the input matroids are given via matrix representations.
    Although the properties of the resulting weight-splitting $c_1$ and $c_2$ of the weight vector $c$ are not stated explicitly, its properties are used in the correctness proofs.
    Properties $(a)$ and $(b)$ can be found in the proof of Theorem 3, and properties $(c)$ and $(d)$ can be found in the proof of Lemma 10.
\end{remark}

The following lemma, which is a slightly modified version of~\cite[Lemma 16]{quanrud:LIPIcs.ICALP.2024.118}, provides a way to compute a compact $(1 + \eps)$-approximate dual solution from the weight functions $c_1$ and $c_2$.

\begin{lemma}[A slightly modified version of {\cite[Lemma 16]{quanrud:LIPIcs.ICALP.2024.118}}] \label{lem:extract_compact_sol}
    Let $I \in \I_1 \cap \I_2$ and let $c_1, c_2 : V \to \mathbb{R}_{\geq 0}$ be functions satisfying the conditions $(a)$--$(d)$ in Theorem~\ref{thm:hkk_weighted}.
    For $t \in \mathbb{R}_{\geq 0}$, define
    \begin{equation*}
        Y_t = \mathrm{span}_{\M_1}(\{ e \in I \mid c_1(e) \geq t \}) \,\,\,  \text{and} \,\,\, Z_t = \mathrm{span}_{\M_2}(\{ e \in I \mid c_2(e) \geq t \}).
    \end{equation*}
    Let $y, z : 2^{V} \to \mathbb{R}_{\geq 0}$ be defined by
    \begin{equation*}
        y = \int_{0}^{\infty} 1_{Y_t} dt \,\,\,  \text{and} \,\,\, z = \int_{0}^{\infty} 1_{Z_t} dt,
    \end{equation*}
    where $1_S$ denote the indicator vector for $S \subseteq V$ in $\mathbb{R}^{2^V}$.
    Then, $y$ and $z$ satisfy the following properties:
    \begin{enumerate}[(i)]
        \item $(y, z)$ are feasible for the dual LP (\ref{equ:dual_lp}).
        \item $(y, z)$ are compact.
        \item $\sum_{S \subseteq V} (\mathrm{rank}_{\M_1}(S) y(S) + \mathrm{rank}_{\M_2}(S) z(S)) \leq (1 + \varepsilon)c(I)$. 
    \end{enumerate}
\end{lemma}

For completeness, we present a proof of Lemma~\ref{lem:extract_compact_sol}, which is almost the same as the proof of \cite[Lemma 16]{quanrud:LIPIcs.ICALP.2024.118}.

\begin{proof}
Here, $y$ (resp. $z$) is supported on a chain of subsets formed by listing the elements of $I$ in decreasing order of $c_1$ (resp. $c_2$).
Thus, the property $(ii)$ holds.

Next, we show that $(y, z)$ is feasible for the dual LP (\ref{equ:dual_lp}).
Fix $e \in V$.
Since $I$ is $c_1$-maximum in $\M_1$, $e \in \mathrm{span}(Y_t)$ for all $t \leq c_1(e)$.
Thus, 
\begin{equation*}
\sum_{S \ni e} y(S) = \int_{0}^{c_1(e)} 1 dt = c_1(e).
\end{equation*}
Symmetrically, $\sum_{S \ni e} z(S) = c_2(e)$, therefore, we have
\begin{equation*}
\sum_{S \ni e} (y(S) + z(S)) = c_1(e) + c_2(e) \geq c(e).
\end{equation*}

It remains to show that the objective value is within a $(1 + \eps)$-factor of $c(I)$.
Here, we have
\begin{align*}
\sum_{S \subseteq V} \mathrm{rank}_{\M_1}(S) y(S) &= \int_0^{\infty} \mathrm{rank}_{\M_1}(Y_t) dt \\
 &= \int_{0}^{\infty} |Y_t \cap I| dt = c_1(I).
\end{align*}
Symmetrically, $\sum_{S \subseteq V} \mathrm{rank}_{\M_2}(S) z(S) = c_2(I)$, therefore, we have
\begin{equation*}
\sum_{S \subseteq V} (\mathrm{rank}_{\M_1}(S) y(S) + \mathrm{rank}_{\M_2}(S) z(S)) = c_1(I) + c_2(I) \leq (1 + \eps)c(I),
\end{equation*}
which completes the proof. 
\end{proof}

By combining Theorem~\ref{thm:hkk_weighted} with Lemma~\ref{lem:extract_compact_sol}, we can compute a $(1 - \eps)$-approximate maximum-weight common independent set $I$, along with a compact $(1 + \eps)$-approximate dual solution $y, z : 2^{[n]} \to \R_{\geq 0}$, with high probability in $\TO(nr^{\omega - 1})$ time, which completes the proof of Lemma~\ref{lem:weighted_sparsified_instance}.

\end{document}